\documentclass{article}
\usepackage{amsthm,amsmath,amssymb,sidecap,subfig,graphicx,hyperref}
\usepackage[left=2.5cm,top=3cm,right=3cm]{geometry}
{
\theoremstyle{definition}

\newtheorem{defi}{Definition}
}

\def\cutrev{cut-stretch}
\newcommand\cut[3]{cut_{#1}(#2,#3)}
\newcommand\comcut[2]{cut_{#1}(#2)}
\newcommand\comcuti[2]{cut^{-1}_{#1}(#2)}
\newcommand\set[1]{\{#1\}}

\newcommand\indegree[2]{deg_{#1}^{-}(#2)}
\newcommand\outdegree[2]{deg_{#1}^{+}(#2)}
\newcommand\reaches[3]{r_{#1}(#2,#3)}

\newcommand\nofedges[3]{d_{#1}(#2,#3)}
\newcommand\recs[1]{{\rm REC(#1)}}
\newcommand\edgerem[2]{#1_{\backslash #2}}

\newcommand\pre[2]{\texttt{pref}(#1,#2)}
\newcommand\inter[2]{\texttt{inter}(#1,#2)}
\newtheorem{lem}{Lemma}
\newtheorem{theo}{Theorem}

\newtheorem{prop}{Proposition}
\newcommand\cutrevf[3]{Cs_{#1}(#2,#3)}

\graphicspath{%
    {Images/}
}

\begin{document}
\title{Feedback arc set problem and NP-hardness of minimum recurrent configuration problem of Chip-firing game on directed graphs\thanks{This paper was partially sponsored by Vietnam Institute for Advanced Study in Mathematics (VIASM) and the Vietnamese National Foundation for Science and Technology Development (NAFOSTED)}}         
\author{K\'evin Perrot and Trung Van Pham}

\date{\today}          
\maketitle
\begin{abstract}
In this paper we present further studies of recurrent configurations of Chip-firing games on Eulerian directed graphs (simple digraphs), a class on the way from undirected graphs to general directed graphs. A computational problem that arises naturally from this model is to find the minimum number of chips of a recurrent configuration, which we call the \emph{minimum recurrent configuration} (MINREC) problem. We point out a close relationship between MINREC and the minimum feedback arc set (MINFAS) problem on Eulerian directed graphs, and prove that both problems are NP-hard.\\
\text{}\\
\noindent\textbf{Keywords.} Chip-firing game, critical configuration, recurrent configuration, Eulerian digraph, feedback arc set, complexity, Sandpile model.
\end{abstract}

\section{Introduction}

A \emph{feedback arc set} of a directed graph (digraph) $G$ is a subset $A$ of arcs of $G$ such that removing $A$ from $G$ leaves an acyclic graph. The \emph{minimum feedback arc set} (MINFAS) problem is a classical combinatorial optimization on graphs in which one tries to minimize $|A|$. This problem has a long history and its decision version was one of Richard M. Karp's 21 NP-complete problems \cite{Kar72}. The problem is known to be still NP-hard for many smaller classes of digraphs such as tournaments, bipartite tournaments, and Eulerian multi-digraphs \cite{CTY07,Fli11,GHM07}. We will prove that it is also NP-hard on Eulerian digraphs, a class in-between undirected and digraphs, in which the in-degree and the out-degree of each vertex are equal.

Chip-firing game is a discrete dynamical system that has received a great attention in recent years, with many variants. The model is a kind of diffusion process on graphs that can be defined informally as follows. Each vertex of a graph has a number of chips and it can give one chip to each of its out-neighbors if it has as many chips as its out-degree. A distribution of chips on the vertices of the graph is called a \emph{configuration}. The model has several equivalent definitions \cite{BTW87, Dha90,BLS91,BL92}. In this paper we refer to the definition that is defined on digraphs by A. Bj\"orner, L. Lov\'asz, and W. Shor \cite{BLS91}. The most important property of Chip-firing games is that if the game converges, it always converges to a unique stable configuration. This property leads to some research directions. A natural direction is the classification of all lattices generated by the converging games \cite{LP01,Mag03}. Most recently, the authors of \cite{PP13} gave a criterion that provides an algorithm for deciding that class of lattices. In this paper we pay attention to another important direction initiated in a paper of N. Biggs. The author defined a variant of Chip-firing game on undirected graphs, the \emph{Dollar game} \cite{Big99}, and studied some special configurations called \emph{critical configurations}. A generalization to the case of digraphs was given in \cite{Dha90,HLMPPW08} where the authors defined \emph{recurrent configurations} and presented many properties that are similar to those of critical configurations on undirected graphs. Holroyd et al. in \cite{HLMPPW08} also studied the Chip-firing game on Eulerian digraphs and presented several typical properties that can also be considered as natural generalizations of the undirected case. In this paper we continue this work and present generalizations of more surprising properties.

A typical property of recurrent configurations is that any stable configuration being component-wise greater than a recurrent configuration is also a recurrent configuration. If the set of minimal recurrent configurations are known, one knows the set of all recurrent configurations. Hence it is worth studying properties of such recurrent configurations. It turns out from the study in \cite{Sch10} that we can associate a minimal recurrent configuration of an undirected graph $G$ with an acyclic orientation of $G$. The  acyclic orientations of $G$ have the same number of arcs, namely $|E(G)|$, so do the total number of chips of minimal recurrent configurations. A direct consequence of this fact is that we can compute the minimum total number of chips of a recurrent configuration in polynomial time since we can compute easily a minimal recurrent configuration. It is natural to ask whether this problem can be solved in polynomial time for the case of digraphs. We will see that the problem becomes much harder than in the undirected case, even when the game is restricted to Eulerian digraphs with a sink. By giving the notion of maximal acyclic arc sets that can be regarded as a generalization of acyclic orientations of undirected graphs, we generalize the definitions and the results in \cite{Sch10} to the class of Eulerian digraphs. Although natural, these generalizations are not easy to see from the studies on undirected graphs. They allow us to derive a number of interesting properties of feedback arc sets and recurrent configurations of the Chip-firing game on Eulerian digraphs, and provide a polynomial reduction from  the MINREC problem to the MINFAS problem on Eulerian digraphs. We extend a result of \cite{Fli11} and show that the MINFAS problem on Eulerian digraphs is also NP-hard, which implies the NP-hardness of the MINREC problem on general digraphs.


The paper is divided into two main sections. The first is devoted to the study of properties of the maximal acyclic arc sets that are complements of the feedback arc sets of an Eulerian digraph. The main result of this section is that finding an acyclic arc set of maximum size can be restricted to looking within particular subsets of acyclic arc sets. By using this result we prove that the MINFAS problem on Eulerian digraphs is NP-hard. It also gives a connection between the MINFAS problem and the MINREC problem on Eulerian digraphs that is presented in the second section. A direct consequence of this connection is the NP-hardness of the MINREC problem on general digraphs.

\section{Acyclic arc sets on Eulerian digraphs}

Throughout this paper a graph always means a simple connected digraph. All results in this paper can be generalized easily to the case of multi-graphs. Traditionally, the vertex set and the edge set of a graph $G$ are denoted by $V(G)$ and $E(G)$, respectively. An \emph{Eulerian digraph} is a digraph in which the in-degree and the out-degree of each vertex are equal. An undirected graph is considered as a digraph in which for any edge linking $u$ and $v$, we consider two arcs: one from $u$ to $v$ and another from $v$ to $u$. With this convention an undirected graph is an Eulerian digraph.  

Let $G=(V,E)$ be a digraph. For a subset $A$ of $E$ let $G[A]$ denote the graph $(V',E')$ with $V'=V$ and $E'=A$. A \emph{feedback arc set} $F$ of $G$ is a subset of $E$ such that removing the arcs in $F$ from $G$ leaves an acyclic graph. An \emph{acyclic arc set} $A$ of $G$ is a subset of $E$ such that the graph $G[A]$ is acyclic. Clearly, an acyclic arc set is the complement of a feedback arc set. A feedback arc set (resp. acyclic arc set) is \emph{minimum} (resp. \emph{maximum}) if it has minimum (resp. maximum) number of arcs over all feedback arc sets (resp. acyclic arc sets) of $G$. A feedback arc set $A$ (resp. acyclic arc set $A$) is \emph{minimal} (resp. \emph{maximal}) if for any $e \in A$ (resp. $e \in E\backslash A$) we have $A\backslash \set{e}$ (resp. $A\cup \set{e}$) is not a feedback arc set (resp. acyclic arc set).


From now until Proposition \ref{cardinality-independence} we work with an Eulerian connected digraph $G=(V,E)$ (note that a connected Eulerian digraph is also strongly connected). A lot of properties of the acyclic arc sets of $G$ are given in this section. The most important result is that finding a maximum acyclic arc set can be restricted to finding an acyclic arc set of the maximum size that has some special properties. This establishes a relation between the MINFAS problem and the MINREC problem on Eulerian digraphs, that we explore in the next section.

For two subsets $A$ and $B$ of $V$, we denote by $\cut{G}{A}{B}$ the set $\set{ (u,v)\in E: u\in A \text{ and } v\in B}$. We write $\comcut{G}{A}$ for $\cut{G}{A}{V\backslash A}$, and $\comcuti{G}{A}$ for $\cut{G}{V\backslash A}{A}$. The following appears stronger than the property $\forall v \in V, \indegree{G}{v}=\outdegree{G}{v}$, but are actually equivalent

\begin{lem}
\label{cut-equation}
For every $A \subseteq V$ we have $|\comcut{G}{A}|=|\comcuti{G}{A}|$.
\end{lem}
\begin{proof}
Let $X=\set{(u,v) \in E: v\in A},Y=\set{(u,v) \in E: u \in A},Z=\set{(u,v) \in E: u\in A\text{ and } v \in A}$. We have $X=\comcuti{G}{A}\cup  Z$ and $Y=\comcut{G}{A}\cup Z$. Since $\comcut{G}{A},\comcuti{G}{A}$ and $Z$ are pairwise disjoint, $|X|=|\comcuti{G}{A}|+|Z|$ and $|Y|=|\comcut{G}{A}|+|Z|$. Since $G$ is Eulerian, we have $0=\underset{v \in A}{\sum}(\indegree{G}{v}-\outdegree{G}{v})=|X|-|Y|=|\comcuti{G}{A}|-|\comcut{G}{A}|$.
\end{proof}

\begin{defi}
Let $A$ be an acyclic arc set and $s$ a vertex of $G$. Let $\reaches{G}{A}{v}$ denote the subset of all vertices of $G$ that are reachable from $s$ by a path in $G[A]$. The set $A \backslash \comcuti{G}{\reaches{G}{A}{s}} \cup \comcut{G}{\reaches{G}{A}{s}}$ is called \emph{\cutrev} of $A$ at $s$. We denote this set by $\cutrevf{G}{A}{s}$.
\end{defi}

\noindent The idea of {\cutrev} is to construct a new acyclic arc set, so that it does not contain less arcs than the old one. Moreover, the number of vertices, that are reachable from a fixed vertex, increases after performing the {\cutrev}. For an intuitive illustration of this definition let us give here an example. Figure \ref{fig:im0} shows an Eulerian digraph with an acyclic arc set $A$ shown in Figure \ref{fig:im1} (plain arcs). If we want to compute the {\cutrev} of $A$ at $v_4$, we look at all vertices reachable from $v_4$ in $G[A]$. These vertices are the set $\reaches{G}{A}{v_4}$ drawn in black on Figure \ref{fig:im2}. The plain arcs in Figure \ref{fig:im3} form the set $\comcuti{G}{\reaches{G}{A}{v_4}}$: arcs of $A$ going from the outside (the set $\set{v_2,v_3,v_7}$) to $\reaches{G}{A}{v_4}$; and the other dotted arcs in this figure form  the set  $\comcut{G}{\reaches{G}{A}{v_4}}$ : arcs of $G$ going from $\reaches{G}{A}{v_4}$ to the outside. Remove the plain arcs in $A$ from $A$ and add the dotted arcs of Figure \ref{fig:im3}, we obtain $\cutrevf{G}{A}{v_4}$ that is shown in Figure \ref{fig:im4}.

\begin{figure}[!h]
\centering
\subfloat[An Eulerian digraph]{\label{fig:im0} \includegraphics[bb=0 0 258 198,width=1.71in,height=1.31in,keepaspectratio]{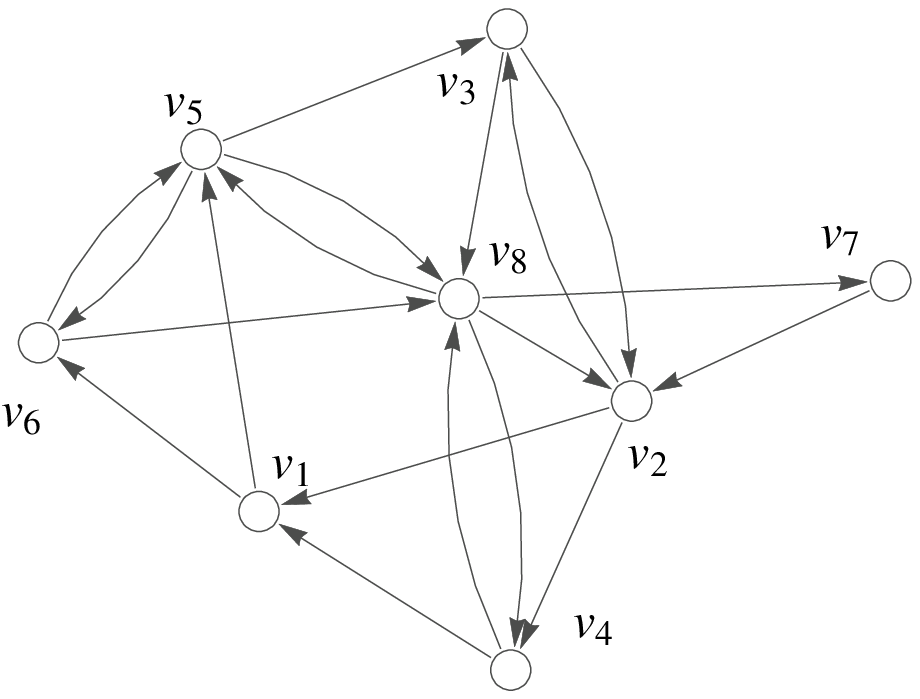} }\\
\subfloat[An acyclic arc set $A$]{\label{fig:im1} \includegraphics[bb=0 0 268 207,width=2.05in,height=1.58in,keepaspectratio]{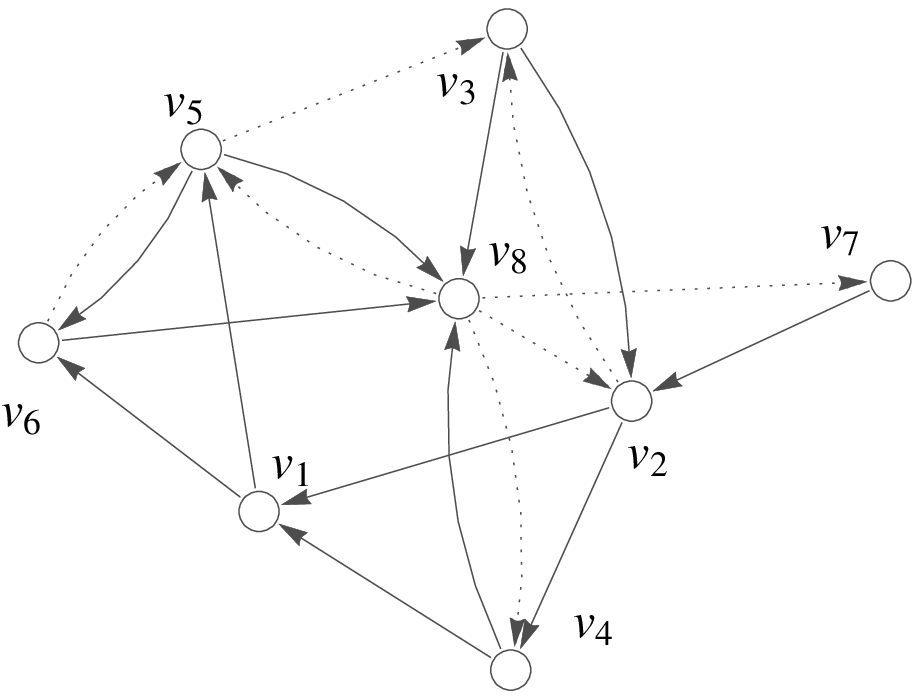}}
\quad
\subfloat[$v_4$ is chosen and the set $R$ of vertices reachable from $v_4$ in $G\text{[}A\text{]}$]{\label{fig:im2} \includegraphics[bb=0 0 303 234,width=2.05in,height=1.58in,keepaspectratio]{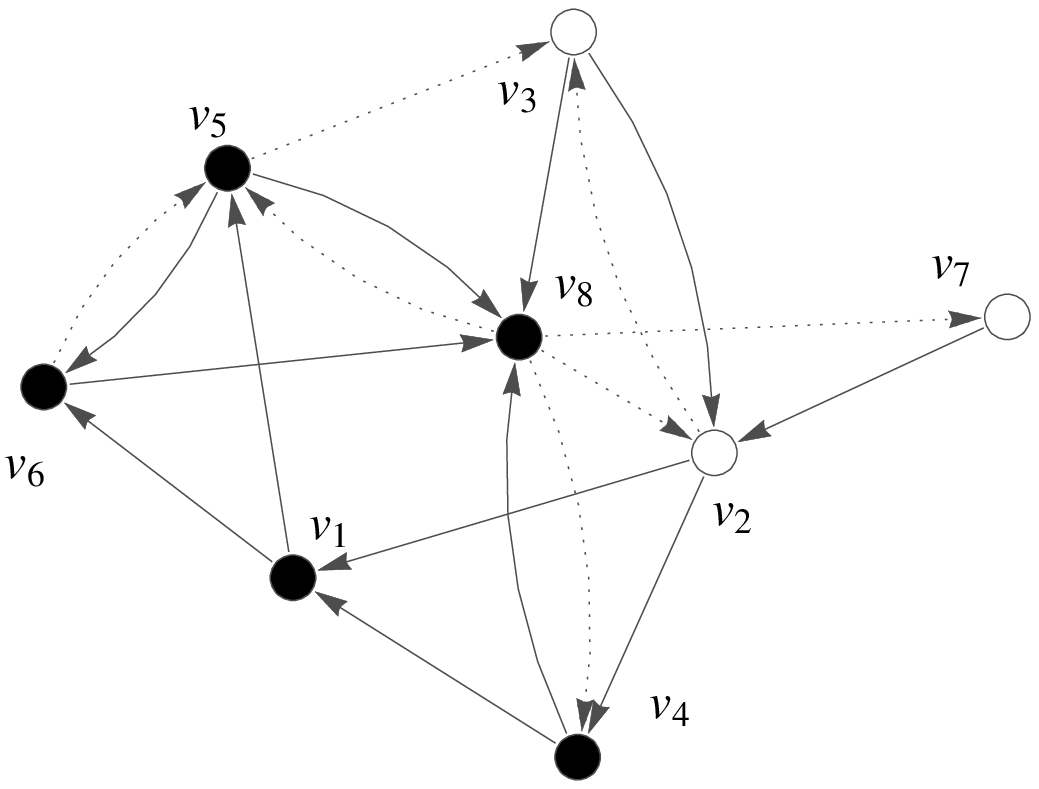}}\\
\subfloat[the arcs of $G$ going into $R$ from the outside (plain arcs) and the arcs of $G$ going from $R$ to the outside (the dotted arcs)]{\label{fig:im3} \includegraphics[bb=0 0 260 200,width=2.05in,height=1.57in,keepaspectratio]{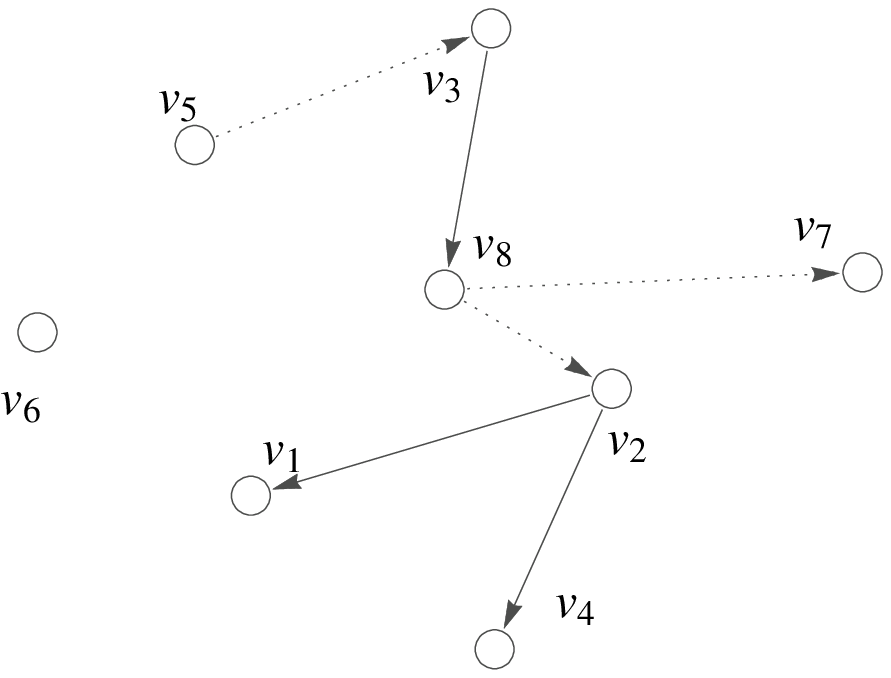} }
\subfloat[The {\cutrev} $\cutrevf{G}{A}{v_4}$]{\label{fig:im4} \includegraphics[bb=0 0 268 207,width=2.05in,height=1.58in,keepaspectratio]{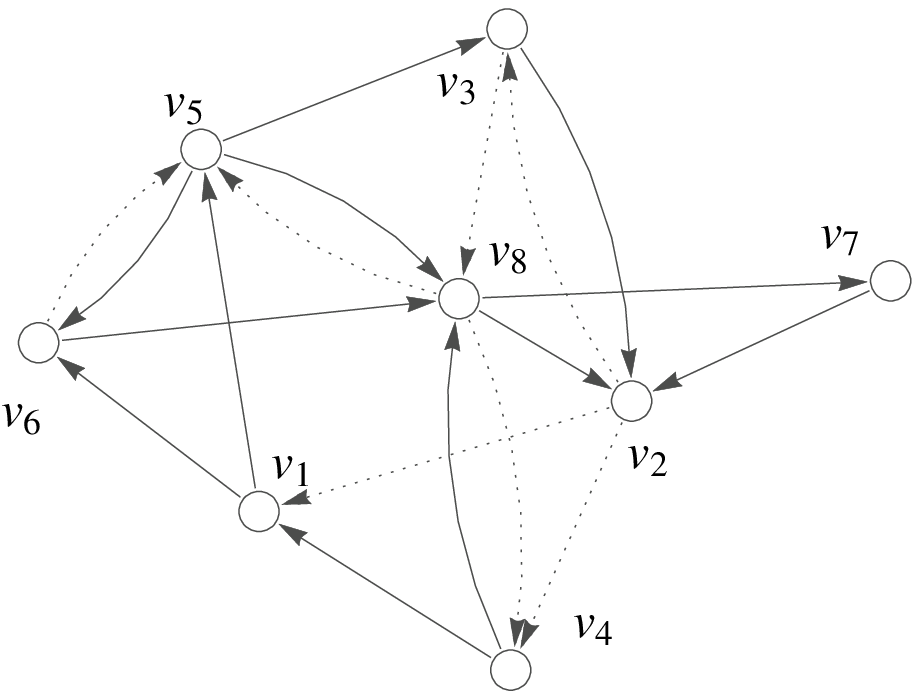} }
\caption{An example of {\cutrev}}
\label{fig:im01234}
\end{figure}

A simple observation from the above example is that a {\cutrev} is still an acyclic arc set and its number of arcs is not less than the number of arcs of the old one. The following shows that this property holds not only for this example but also holds for the general case.
\begin{lem}
\label{cardinality non-decreasing}
Let $A$ be an acyclic arc set and $s$ a vertex of $G$. Then $\cutrevf{G}{A}{s}$ is also an acyclic arc set of $G$. Moreover $|A|\leq |\cutrevf{G}{A}{s}|$.
\end{lem}
\begin{proof}
By the definition of {\cutrev} there is no arc in $\cutrevf{G}{A}{s}$ from a vertex in $V\backslash \reaches{G}{A}{s}$ to a vertex in $\reaches{G}{A}{s}$. It implies that if $\cutrevf{G}{A}{s}$ contains a cycle, the vertices in this cycle must be completely contained either in $\reaches{G}{A}{s}$ or in $V\backslash \reaches{G}{A}{s}$. In this case the arcs of the cycle are also the arcs of $A$, therefore the cycle is also a cycle of $A$, a contradiction to the acyclicity of $A$.

To prove $|A|\leq |\cutrevf{G}{A}{s}|$, we observe that $A \cap \comcut{G}{\reaches{G}{A}{s}}=\emptyset$ (from the maximality of $\reaches{G}{A}{s}$). 
From Lemma \ref{cut-equation}, we have $|\cutrevf{G}{A}{s}|\geq |A|+|\comcut{G}{\reaches{G}{A}{s}}|-|\comcuti{G}{\reaches{G}{A}{s} }|=|A|$, which completes the proof.
\end{proof}

The following is the main result of this subsection.

\begin{theo}
\label{unique-source property}
Let $N$ be the maximum number of arcs of an acyclic arc set of $G$. For every vertex $s$ of $G$ there is an acyclic arc set of $N$ arcs such that it contains no arc whose head is $s$.
\end{theo}

\begin{proof}
Let $X$ be an acyclic set of $G$ of $N$ arcs. We construct a sequence $\set{A_i}_{i\in \mathbb{N}}$ as follows:
 $A_0=X$ and $A_i=\cutrevf{G}{A_{i-1}}{s}$ for every $i\geq 1$. Lemma \ref{cardinality non-decreasing} and the maximum of $N$ imply that $|A_i|=N$ for every $i \in \mathbb{N}$. If $\reaches{G}{A_k}{s}=V$ for some $k$, $A_k$ is an acyclic set that has the required property since for any vertex $v \neq s$ of $G$ the existence of a path in $A_k$ from $s$ to $v$ implies that $(v,s) \not \in A_k$. The proof is completed by showing that  there always exists such a $k$.

Since a path from $s$ in $G[A_i]$ is also a path from $s$ in $G[A_{i+1}]$, we have $\reaches{G}{A_i}{s}\subseteq \reaches{G}{A_{i+1}}{s}$.  It suffices to show that if $\reaches{G}{A_i}{s}\subsetneq V$ then $\reaches{G}{A_i}{s}\subsetneq \reaches{G}{A_{i+1}}{s}$. Since $\reaches{G}{A_i}{s}\subsetneq V$, there is an arc $e=(v_1,v_2)$ of $G$ such that $v_1\in \reaches{G}{A_i}{s}$ and $v_2 \not \in \reaches{G}{A_i}{s}$. Since $e \in A_{i+1}$, there is a path in $A_{i+1}$ that is from $s$ to $v_2$ going through $v_1$. It implies that $v_2\in \reaches{G}{A_{i+1}}{s}$, therefore $\reaches{G}{A_i}{s}\subsetneq \reaches{G}{A_{i+1}}{s}$.
\end{proof}

\begin{defi}
A vertex $s$ of $G$ is called \emph{sink} in an acyclic arc set $A$ of $G$ if $s$ has indegree $0$ in $G[A]$. Let $A$ be an acyclic arc set of $G$ such that $A$ has exactly one sink $s$. A vertex $s'$ of $G$ distinct from $s$ is called \emph{sinkable} in $A$ if there is an arc of $G$ whose head is $s$ and whose tail is in $\reaches{G}{A}{s'}$.
\end{defi}

\noindent We call such a vertex $s'$ sinkable because the idea is to use the arc from $s'$ to $s$ to construct an acyclic arc set where it becomes a sink. The fact that this is done by the {\cutrev} at $s'$ is stated in the following lemma.

\begin{lem}
\label{sinkability}
Let $A$ be an acyclic arc set of $G$ having exactly one sink $s$. If $s'$ is sinkable in $A$ then $\cutrevf{G}{A}{s'}$ has exactly one sink $s'$ and $s$ is sinkable in $\cutrevf{G}{A}{s'}$. Moreover $A\subseteq \cutrevf{G}{\cutrevf{G}{A}{s'}}{s}$.
\end{lem}
\begin{proof}
Let $X$ denote $\reaches{G}{A}{s'}$ and $Y=V\backslash X$. Since $A$ has exactly one sink $s$, for any $v \in V$ there is a path in $A$ from $s$ to $v$, therefore from $s$ to $s'$. The acyclicity of $G[A]$ implies that $s \in Y$.

Clearly, $s'$ is a sink of $\cutrevf{G}{A}{s'}$. To prove that $\cutrevf{G}{A}{s'}$ has a unique sink, it suffices to show that for any $v \in V$ there is a path in $\cutrevf{G}{A}{s'}$ from $s'$ to $v$. It is trivial if $v \in X$. We consider the case $v \in Y$. Let $(v',s)$ be an arc of $G$ such that $v'\in X$. Such an arc exists because of the assumption of the lemma. By the definition of {\cutrev} we have $(v',s)\in \cutrevf{G}{A}{s'}$. Let $P_1$ and $P_2$ be paths in $A$ from $s'$ to $v'$ and from $s$ to $v$, respectively. It follows from the definition of {\cutrev} that $\cutrevf{G}{A}{s'}$ contains $P_1$. Since $v \in Y$, the path $P_2$ goes through only the vertices in $Y$. Therefore $\cutrevf{G}{A}{s'}$ also contains $P_2$. Hence the path $P_1\cup \set{(v',s)}\cup P_2$ is a path in $\cutrevf{G}{A}{s'}$ from $s'$ to $v$.

Let $P_3$ be a path in $A$ from $s$ to $s'$. The acyclicity of $G[A]$ implies that $P_3$ goes through only the vertices in $Y\cup \set{s'}$. Therefore there is an arc $(v'',s')$ such that $v''\in Y$. Clearly, we have $\reaches{G}{\cutrevf{G}{A}{s'}}{s}=Y$. By the definition of the sinkability we have $s$ is sinkable in $\cutrevf{G}{A}{s'}$.

It remains to show that $A\subseteq \cutrevf{G}{\cutrevf{G}{A}{s'}}{s}$. This follows immediately from the fact that\linebreak $\cutrevf{G}{\cutrevf{G}{A}{s'}}{s}=A\cup \cut{G}{Y}{X}$.
\end{proof}

For each $s \in V$, let $\chi_{s}$ denote the number of maximum acyclic arc sets of $G$ with exactly one sink $s$. It is well-known that for an undirected graph $G$, $T_G(1,0)$ counts the number of acyclic orientations with a unique fixed sink, therefore counts $\chi_{s}$, where $T_G(x,y)$ is the Tutte polynomial of $G$. This implies that if $G$ is an undirected graph, $\chi_{s}$ is independent of the choice of $s$. The following is a generalization of this fact to Eulerian digraphs

\begin{prop}
\label{cardinality-independence}
For any two vertices $s_1,s_2$ of $G$ we have $\chi_{s_1}=\chi_{s_2}$.
\end{prop}

\begin{proof}
We claim that if $(v',v) \in E(G)$ then $\chi_{v}\leq \chi_{v'}$. Let $\mathcal{A}_1$ denote the set of maximum acyclic arc sets of $G$ having exactly one sink $v$, and $\mathcal{A}_2$ the set of maximum acyclic arc sets having exactly one sink $v'$. Since $(v',v) \in E(G)$, $v'$ is sinkable in every acyclic arc set in $\mathcal{A}_1$. It follows from Theorem \ref{unique-source property} and Lemma \ref{sinkability} that the the map $\theta:\mathcal{A}_1\to \mathcal{A}_2$, defined by $A\to \cutrevf{G}{A}{v'}$, is well-defined. Let $A$ be arbitrary in $\mathcal{A}_1$. It follows from Lemma \ref{sinkability} that $A\subseteq \cutrevf{G}{\cutrevf{G}{A}{v'}}{v}$. Since $A$ is maximum, we have $A= \cutrevf{G}{\cutrevf{G}{A}{v'}}{v}$. This implies that $\theta$ is injective. Therefore $|\mathcal{A}_1|\leq |\mathcal{A}_2|$, equivalently $\chi_{v}\leq \chi_{v'}$.

The claim implies that for any two vertices $v'$ and $v$ of $G$ such that there is a path in $G$ from $v'$ to $v$, we have $\chi_{v}\leq \chi_{v'}$. Since $G$ is strongly connected, there is a path in $G$ from $s_1$ to $s_2$ and a path in $G$ from $s_2$ to $s_1$. Hence $\chi_{s_1}=\chi_{s_2}$.
\end{proof}

Note that in an undirected graph a maximal acyclic arc set is also a maximum acyclic arc set (and vice versa). This fact no longer holds for Eulerian digraphs. The assertion in Proposition \ref{cardinality-independence} is not correct if we replace the maximum acyclic arc sets by the maximal acyclic arc sets.

We recall the definition of the MINFAS problem
\begin{center}
\begin{tabular}{|l|}
\hline
\textbf{MINFAS Problem}\\
\text{}\\
\textbf{Input:} A digraph $G$\\
\textbf{Output:} Minimum number of arcs of a feedback arc set of $G$\\
\hline
\end{tabular}
\end{center}
When the problem is restricted to Eulerian digraphs, we call it \emph{EMINFAS} problem for short. Although the EMINFAS problem was known to be NP-hard for its multigraph version \cite{Fli11}, it is worth studying the computational complexity of the EMINFAS problem since most variants of the MINFAS problem are restrictions of the class of digraphs (simple) (see \cite{HMSSY12}). It does not seem that the construction in the proof of \cite{Fli11} is applicable to the case of simple digraphs. By using Theorem \ref{unique-source property} and a stronger construction we show that the EMINFAS is NP-hard. We work with a general digraph $G=(V,E)$, and construct an Eulerian digraph $G'$ so that an optimum value of the EMINFAS problem on $G'$ implies an optimum value of the MINFAS problem on $G$. The graph $G'=(V',E')$ used in the reduction of MINFAS to EMINFAS is constructed as follows.

The basic idea to construct an Eulerian graph $G'$ from $G$ would be to create a new vertex and add arcs from this new vertex to any vertex that has more out-degree than in-degree, and arcs from vertices which have in-degree greater than out-degree to the new vertex. To avoid multi-graphs, we furthermore add for each of those arcs a new vertex in between, which has in-degree and out-degree $1$. More precisely, the vertices of $G$ are denoted by $v_1,v_2,\cdots,v_n$ for some $n$. If $G$ is already an Eulerian digraph then $G':=G$. Otherwise let $G'$ be a copy of $G$. We add to $G'$ a new vertex $s$. For each vertex $v_i$ such that $\indegree{G}{v_i}<\outdegree{G}{v_i}$ we add $p_i$ new vertices $w_{i,1},w_{i,2},\cdots,w_{i,p_i}$ to $G'$, and for each $j \in [1..p_i]$ we add two arcs $(s,w_{i,j})$ and $(w_{i,j},v_i)$ to $G'$, where $p_i=\outdegree{G}{v_i}-\indegree{G}{v_i}$. For each vertex $v_i$ such that $\outdegree{G}{v_i}<\indegree{G}{v_i}$ we add $q_i$ new vertices $w_{i,1},w_{i,2},\cdots,w_{i,q_i}$ to $G'$, and for each $j \in [1..q_i]$ we add two arcs $(w_{i,j},s)$ and $(v_i,w_{i,j})$ to $G'$, where $q_i=\indegree{G}{v_i}-\outdegree{G}{v_i}$. Formally, the vertex set and the arc set of $G'$ are defined by 
\begin{align*}
V':=\set{s}\cup &V\cup \underset{1 \leq i \leq n}{\bigcup} \set{w_{i,j}: 1 \leq j \leq |\indegree{G}{v_i}-\outdegree{G}{v_i}|}\\
E':=E\cup &\underset{\indegree{G}{v_i}<\outdegree{G}{v_i}}{\bigcup} \set{(s,w_{i,j}): 1 \leq j \leq \outdegree{G}{v_i}-\indegree{G}{v_i}}\cup\\
&\underset{\indegree{G}{v_i}<\outdegree{G}{v_i}}{\bigcup} \set{(w_{i,j},v_i): 1 \leq j \leq \outdegree{G}{v_i}-\indegree{G}{v_i}}\cup\\
&\underset{\outdegree{G}{v_i}<\indegree{G}{v_i}}{\bigcup} \set{(w_{i,j},s): 1 \leq j \leq \indegree{G}{v_i}-\outdegree{G}{v_i}}\cup\\
&\underset{\outdegree{G}{v_i}<\indegree{G}{v_i}}{\bigcup} \set{(v_i,w_{i,j}): 1 \leq j \leq \indegree{G}{v_i}-\outdegree{G}{v_i}}
\end{align*}

\begin{figure}
\centering
\subfloat[A digraph $G$]{\label{fig:im5}\includegraphics[bb=0 0 300 145,width=2.05in,height=0.989in,keepaspectratio]{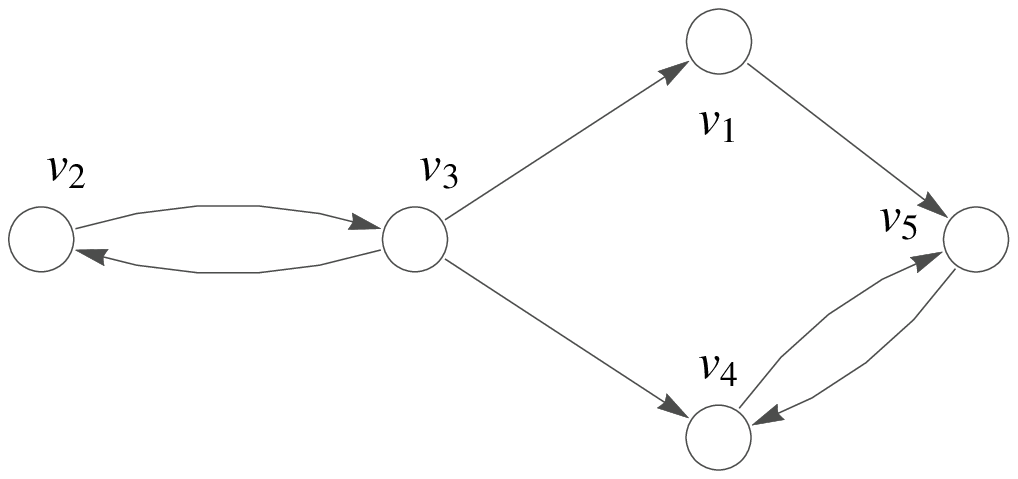}} \quad
\subfloat[Eulerian digraph $G'$]{\label{fig:im6} \includegraphics[bb=0 0 300 207,width=2.05in,height=1.41in,keepaspectratio]{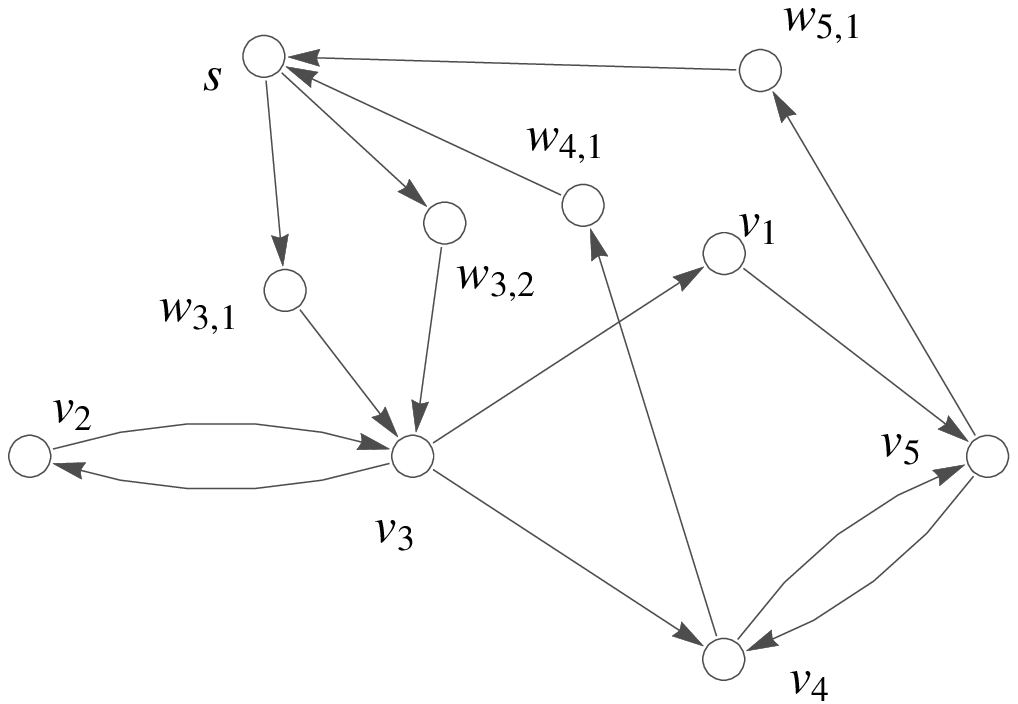}}\\
\subfloat[An acyclic arc set of $G$ of maximum cardinality]{\label{fig:im7} \includegraphics[bb=0 0 300 145,width=2.05in,height=0.989in,keepaspectratio]{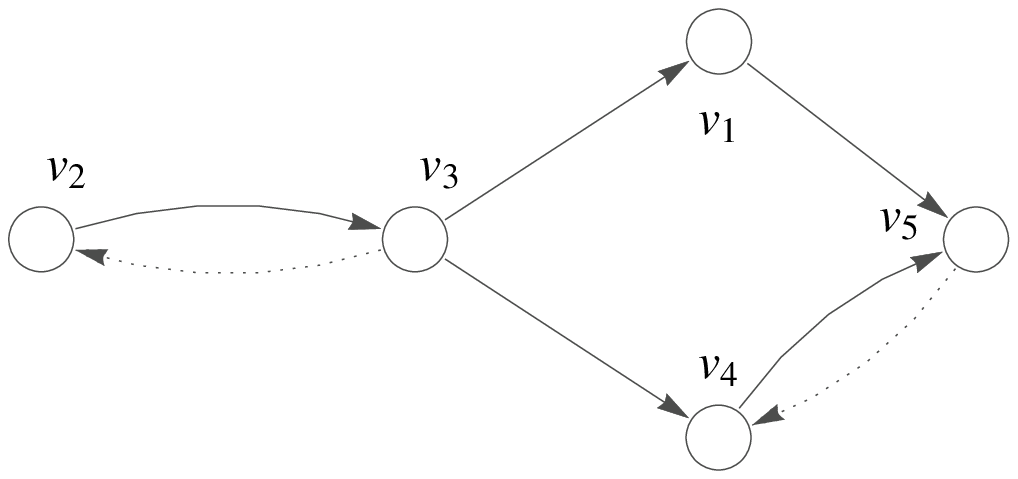} }\quad
\subfloat[An acyclic arc set of $G'$]{\label{fig:im8} \includegraphics[bb=0 0 300 207,width=2.05in,height=1.41in,keepaspectratio]{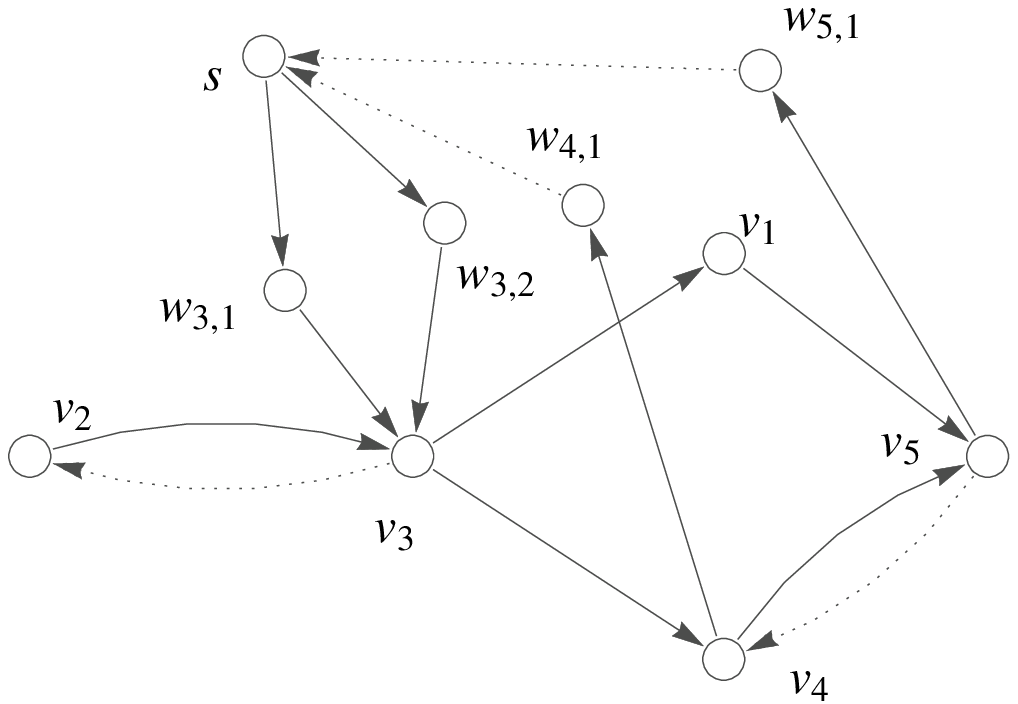}}
\caption{Maximum acyclic arc sets}
\label{fig:im5678}
\end{figure}

Figure \ref{fig:im5678} shows an example of $G$ (Fig. \ref{fig:im5}) and the corresponding Eulerian digraph $G'$ (Fig. \ref{fig:im6}). Figure \ref{fig:im7} shows an acyclic arc set of $G$ of maximum cardinality. In order to construct an acyclic arc set of $G'$, we add the arcs $(s,w_{i,j}),(w_{i,j},v_i)$ (all the arcs created to offset vertices having out-degree greater than in-degree in $G$) and $(v_i,w_{i,j})$ (half of the arcs created to offset vertices having in-degree greater than out-degree in $G$) to this set, which indeed results in an acyclic arc set of $G'$ of maximum cardinality. The following shows that we can always obtain an acyclic arc set of $G'$ of maximum cardinality with this construction.

\begin{lem}
\label{maximum acyclic arc sets relation}
Let $r$ be the maximum number of arcs of an acyclic arc set of $G$, and
$$d=\underset{\indegree{G}{v_i} \,<\, \outdegree{G}{v_i}}{\sum} (\outdegree{G}{v_i}-\indegree{G}{v_i}).$$
The maximum number of arcs of an acyclic arc set of $G'$ is $3d+r$.
\end{lem}

\begin{proof}
The lemma clearly holds if $G$ is an Eulerian digraph, in which case $G'=G$. We assume otherwise. Note that $4d$ arcs and $2d+1$ vertices are added to $G$ in order to construct $G'$. Let $r'$ be the maximum number of arcs of an acyclic arc set of $G'$. 

First, we show that $3d+r\leq r'$. Let $A$ be an acyclic arc set of $G$ of $r$ arcs. Let $A'=A \cup \set{(s,w_{i,j}):(s,w_{i,j})\in E'}\cup \set{(w_{i,j},v_i):(w_{i,j},v_i)\in E'}\cup \set{(v_i,w_{i,j}):(v_i,w_{i,j})\in E'}$. Since $A$ is an acyclic arc set of $G$ and $A'$ contains no arc $(w_{i,j},s)$ of $E'$, $A'$ is an acyclic arc set of $G'$. The sets $\set{(s,w_{i,j}):(s,w_{i,j})\in E'}$,$\set{(w_{i,j},v_i):(w_{i,j},v_i)\in E'}$ and $ \set{(v_i,w_{i,j}):(v_i,w_{i,j})\in E'}$ are pairwise-disjoint, and each of them has exactly $d$ arcs, therefore we have constructed an acyclic arc set $A'$ of size $|A'|=3 d+r$. It implies that $3d+r\leq r'$.

It remains to show that $r'\leq 3d+r$. Let $B$ be an acyclic arc set of $G'$ of $r'$ arcs. By Theorem \ref{unique-source property} there is an acyclic arc set $B'$ of $G'$ of $r'$ arcs such that $B'$ contains no arc $(w_{i,j},s) $ of $E'$. The set $B'$ must contain all arcs $e$ of $G'$ of the form $(s,w_{i,j})$, $(w_{i,j},v_i)$ or $(v_i,w_{i,j})$ since if otherwise, $B'\cup \set{e}$ is an acyclic arc set of $G'$ containing $r'+1$ arcs. Let $A''$ denote $B'\backslash \big(\set{(s,w_{i,j}): (s,w_{i,j}) \in E'}\cup \set{(w_{i,j},v_i):(w_{i,j},v_i)\in E'}\cup \set{(v_i,w_{i,j}):(v_i,w_{i,j})\in E'}\big)$. The set $A''$ is an acyclic arc set of $G$, therefore $|A''|\leq r$. It implies $r'=|B'|=3d+|A''|\leq 3d+r$.
\end{proof}

A direct consequence of Lemma \ref{maximum acyclic arc sets relation} is a NP-hardness proof for the EMINFAS problem.
\begin{theo}
The EMINFAS problem is NP-hard.
\end{theo}
\begin{proof}
Given a general digraph $G$, the Eulerian digraph $G'$ can be constructed in polynomial time. Let $b$ be the minimum number of arcs of a feedback arc set of $G'$, that is, the solution of EMINFAS on $G'$. Clearly $|E'|-b$ is the maximum number of arcs of an acyclic arc set of $G'$. By Lemma \ref{maximum acyclic arc sets relation} the maximum number of arcs of an acyclic arc set of $G$ is $|E'|-b-3 d$, where $d$ is defined as in Lemma \ref{maximum acyclic arc sets relation} and is computable in polynomial time. Thus the minimum number of arcs of a feedback arc set of $G$ is $|E|-(|E'|-b-3d)=b+3d+|E|-|E'|$. This implies a polynomial-time reduction from the MINFAS problem to the EMINFAS problem. The MINFAS problem is NP-hard, so is the EMINFAS problem.
\end{proof}

\section{NP-hardness of minimum recurrent configuration problem}

\subsection{Chip-firing game}

\subsubsection{Chip-firing game on digraphs} 

Let $G=(V,E)$ be a digraph. A vertex $s$ is called a \emph{global sink} if $\outdegree{G}{s}=0$ and for any $v\in V$ there is a path from $v$ to $s$ (possibly a path of length $0$). Clearly if $G$ has a global sink then it is unique. 

A \emph{configuration} $c$  of $G$ is a map from $V$ to $\mathbb{N}$. The value $c(v)$ can be regarded as the number of chips stored at $v$. A vertex $v$ of $G$ is \emph{active} if $c(v)\geq \outdegree{G}{v}\geq 1$. Configuration $c$ is \emph{stable} if $c$ has no active vertex. \emph{Firing} at $v$ results in the map $c':V\to \mathbb{Z}$ that is defined by 
$$
c'(w)=
\begin{cases}
c(w)-\outdegree{G}{w}&\text{ if }w=v\\
c(w)+1&\text{if } v \neq w \text{ and } (v,w)\in E\\
c(w)&\text{otherwise}
\end{cases}
$$
This firing is often denoted by $c\overset{v}{\to} c'$. Clearly if $v$ is active then $c'$ is also a configuration of $G$. In this case the firing $c\overset{v}{\to}c'$ is called \emph{legal}. If $d$ is obtained from $c$ by a sequence of legal firings (possibly a sequence of length $0$), we write $c\overset{*}{\to}d$.

A game begining with initial configuration $c_0$ and playing with legal firings is called \emph{a Chip-firing game}. Note that at each step of firing there are possibly more than one active vertex, therefore there are possibly more than one choice of legal firing. As a consequence, it may be a complicated problem if one wants to know the termination of the game. Hopefully, it is not the case for the Chip-firing model since the termination has a good characterization.

\begin{lem}\cite{BL92} Let $G$ be a digraph and $c$ an initial configuration. Then the game either plays forever or arrives at a unique stable configuration. Moreover if $G$ has a global sink, the game arrives  at a stable configuration. We denote by $c^{\circ}$ this stable configuration.
\end{lem}

\subsubsection{Recurrent configuration}

Let $G=(V,E)$ be a digraph with global sink $s$. Since $s$ is always not active no matter how many chips it has, it makes sense to define a configuration on $G$ to be a map from $V\backslash \set{s}$ to $\mathbb{N}$. In a firing when a chip goes into $s$, it vanishes. Therefore the total number of chips is no longer an invariant under firings. A configuration $c$ is \emph{accessible} if for any configuration $d$ there is a configuration $d'$ such that $(d+d')\overset{*}{\to}c$, where $d+d'$ is the configuration given by $(d+d')(v)=d(v)+d'(v)$ for any $v \in V\backslash \set{s}$. Configuration $c$  is \emph{recurrent} if it is both stable and accessible. We denote by $\recs{G}$ the set of all recurrent configurations of $G$.

Fix a linear order $v_1<v_2<\cdots<v_n$ on $V$, where $n=|V|$. The \emph{Laplacian matrix} $\Delta$ of $G$ with respect to the order is given by
$$\Delta_{i,j}=
\begin{cases}
\nofedges{G}{v_i}{v_j}&\text{ if } i\neq j\\
-\outdegree{G}{v_i}&\text{ if } i=j
\end{cases}
$$
With the order a configuration can be represented by a vector of $\mathbb{Z}^{n-1}$, therefore can be regarded as an element of the group $(\mathbb{Z}^{n-1},+)$. Let $\Delta_{\backslash s}$ denote the matrix $\Delta$ in which the row and the column corresponding to $s$ have been deleted. We define an equivalence relation $\sim$ on the set of all configurations of $G$ by $c_1\sim c_2$ iff there is a row vector $z \in \mathbb{Z}^{n-1}$ such that $c_1-c_2=z \cdot {\Delta}_{\backslash s}$. The following shows a relation between the set of recurrent configurations and the equivalence classes.

\begin{lem}\cite{HLMPPW08}
The set of all recurrent configurations $\recs{G}$ is an Abelian group with the addition defined by $c\oplus c':=(c+c')^{\circ}$. Moreover, each equivalence class according to $\sim$ contains exactly one recurrent configuration, and $|\recs{G}|$ is equal to the number of the equivalence classes.
\end{lem}

Naturally, one asks if it is possible to verify efficiently whether a given configuration is recurrent? The definition of recurrent configuration does not imply an efficient algorithm for this computational problem. Nevertheless, the following implies a polynomial-time algorithm for this problem.

\begin{lem} \cite{HLMPPW08}
\label{verify a recurrent configuration}
Let $\delta$ be the configuration defined by $\delta(v)=2\outdegree{G}{v}$ for every $v \in V\backslash\set{s}$, and $\epsilon$ be the configuration given by $\epsilon(v)=\delta(v)-\delta^{\circ}(v)$ for every $v\in V\backslash\set{s}$. The configuration $\epsilon$ belongs to the equivalence class of the identity element, and a configuration $c$ is recurrent iff $c=(c+\epsilon)^{\circ}$.
\end{lem}

\noindent Note that the assertion of Lemma \ref{verify a recurrent configuration} still holds if we replace the definition of $\delta$ in the lemma by $\delta(v)=\outdegree{G}{v}$ for every $v \in V\backslash\set{s}$. The following is a generalization of Lemma \ref{verify a recurrent configuration}, where $\textbf{0}$ denotes the zero-configuration, \emph{i.e.} $\textbf{0}(v)=0$ for every $v \in V\backslash \set{s}$ ($\textbf{0}$ is in the equivalence class of the identity, but is not a recurrent configuration).

\begin{lem}
\label{weaker condition}
Let $A$ be a subset of $V\backslash \set{s}$ satisfying that for every $v \in V$ there is a path in $G$ from a vertex in $A$ to $v$. Let $\beta$ be a configuration such that $\beta$ is in the same equivalence class as \textbf{0} and $\beta(v)>0$ for every $v \in A$. Then a configuration $c$ is recurrent iff $c=(c+\beta)^{\circ}$.
\end{lem}

\begin{proof}
$\Rightarrow:$ Let $\bar{c}=(c+\beta)^{\circ}$. The proof is completed by showing that $\bar{c}$ is recurrent. Configuration $\bar{c}$ is stable, therefore it remains to prove that $\bar{c}$ is accessible. Let $d$ be a configuration. Since $c$ is recurrent, there is a configuration $d''$ such that $c=(d+d'')^{\circ}$, therefore $\bar{c}=(c+\beta)^{\circ}=(d+d''+\beta)^{\circ}$. Let $d'=d''+\beta$. We have $\bar{c}=(d+d')^{\circ}$.

$\Leftarrow:$ For $k \in \mathbb{N}$ let $k \beta$ be the configuration defined by $(k\beta)(v)=k \cdot \beta(v)$ for every $v\in V\backslash \set{s}$. Since for every $v \in V$ there is a path from a vertex in $A$ to $v$, with $k$ large enough and by an appropriate  sequence of legal firings the configuration $k \beta$ arrives at a configuration $c'$ that satisfies $c'(v)\geq \outdegree{G}{v}$ for every $v \in V\backslash \set{s}$. We have $c=(c+k \beta)^{\circ}=(c+c')^{\circ }$. Since $(c+c')(v)\geq \outdegree{G}{v}$ for every $v \in V\backslash \set{s}$, $(c+c')^{\circ}$ is accessible,\linebreak so is $c$.
\end{proof}

\begin{figure}
\centering
\subfloat[A digraph with global sink $s$]{\label{fig:im9}\includegraphics[bb=0 0 260 136,width=2.05in,height=1.07in,keepaspectratio]{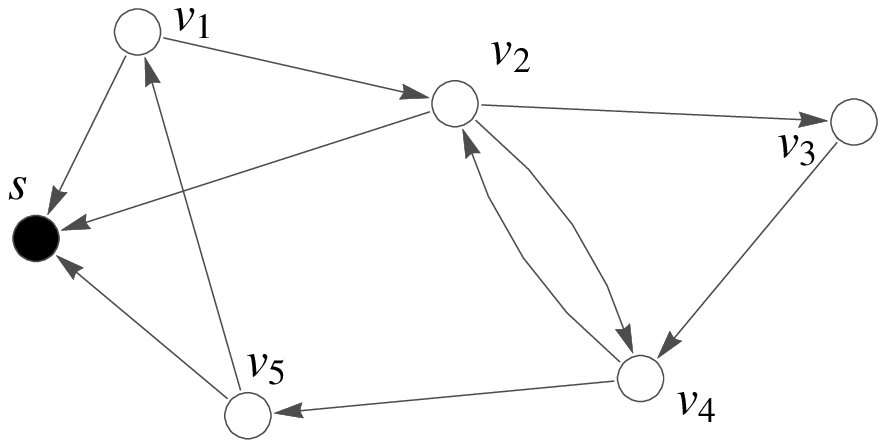}}\quad 
\subfloat[A configuration $c$]{\label{fig:im11(1)}\includegraphics[bb=0 0 348 180,width=2.05in,height=1.06in,keepaspectratio]{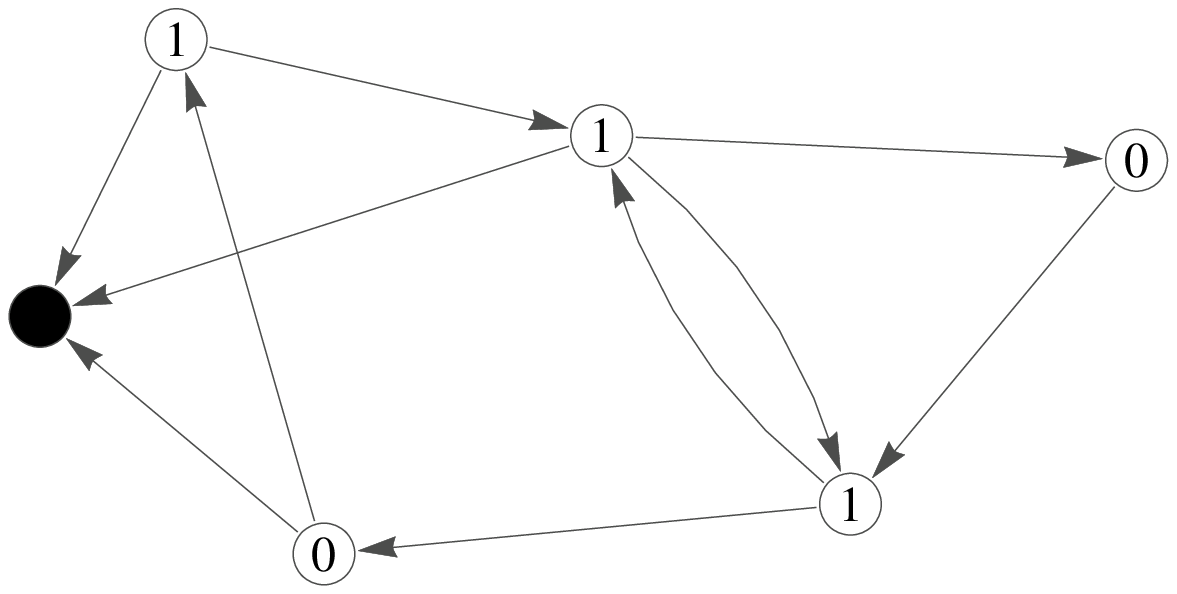}}\\
\subfloat[Configuration $\epsilon$]{\label{fig:im10}\includegraphics[bb=0 0 348 180,width=2.05in,height=1.06in,keepaspectratio]{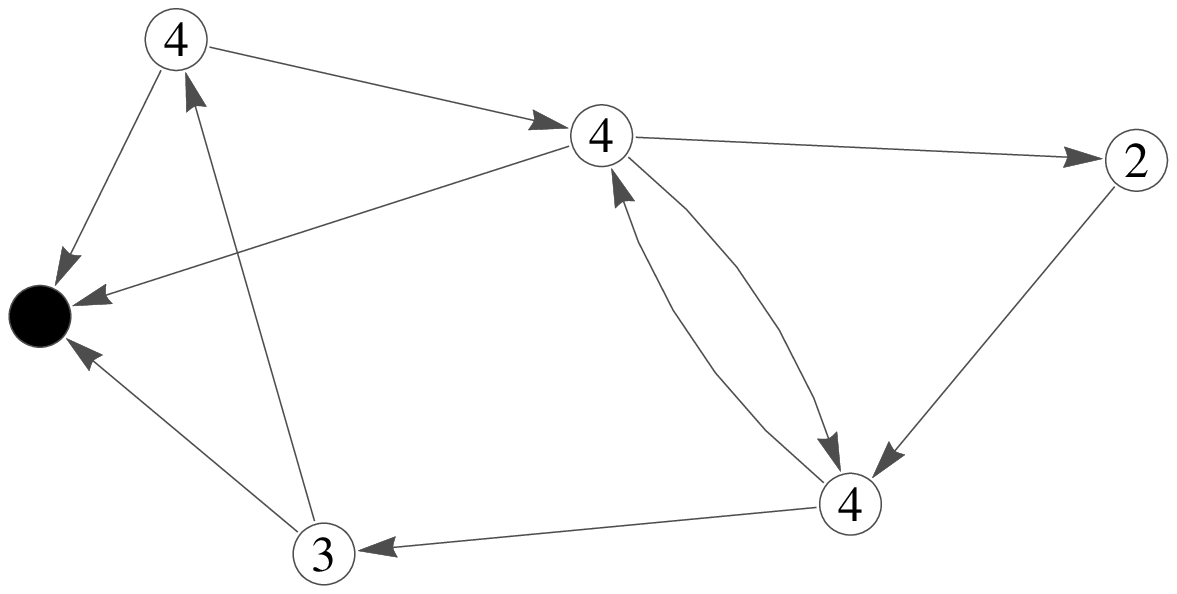}}\quad 
\subfloat[Configuration $\beta$]{\label{fig:im12}\includegraphics[bb=0 0 348 180,width=2.05in,height=1.06in,keepaspectratio]{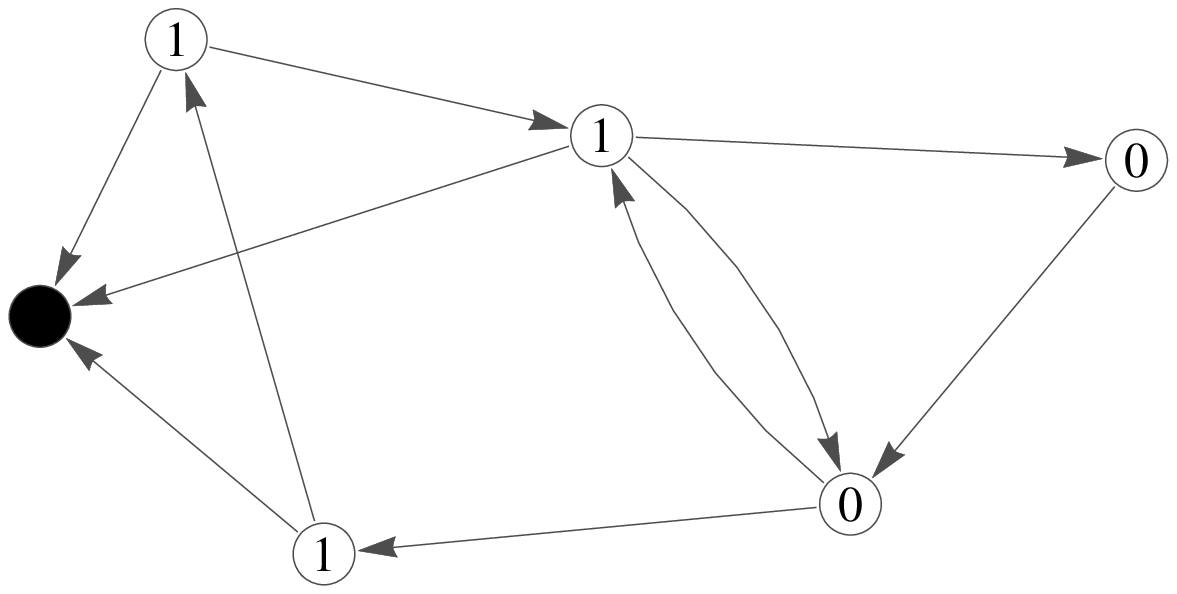}}\\
\subfloat[$(c+\epsilon)^{\circ}$]{\label{fig:im11(2)}\includegraphics[bb=0 0 348 180,width=2.05in,height=1.06in,keepaspectratio]{im11.eps}}\quad
\subfloat[$(c+\beta)^{\circ}$]{\label{fig:im11(3)}\includegraphics[bb=0 0 348 180,width=2.05in,height=1.06in,keepaspectratio]{im11.eps}}
\caption{Verifying a recurrent configuration}
\label{fig:im09101112}
\end{figure}

\noindent Lemma \ref{verify a recurrent configuration} is a special case of Lemma \ref{weaker condition} with $A=V\backslash \set{s}$ and $\beta=\epsilon$.

Figure \ref{fig:im9} shows a digraph with global sink $s$ and a configuration $c$ on the right (Fig. \ref{fig:im11(1)}). If we want to decide whether this configuration is recurrent, we construct the configuration $\epsilon$ (Figure \ref{fig:im10}), and compute the stable configuration $(c+\epsilon)^{\circ}$ (Figure \ref{fig:im11(2)}). Lemma \ref{verify a recurrent configuration} states that $c$ is recurrent if and only if $c=(c+\epsilon)^{\circ}$. However $\epsilon$ has a large number of chips, and the computation of $(c+\epsilon)^{\circ}$ may be long. It may be more time efficient to us another configuration with fewer chips, for example the one given on Figure \ref{fig:im12}, $\beta$, which allows to decide in a similar way if $c$ is recurrent (Lemma \ref{weaker condition} applies since the digraph has a global sink $s$). Consider performing the stabilization by hand: one would clearly prefer using $\beta$ to $\epsilon$.

\subsubsection{Chip-firing game on Eulerian digraphs with sink and firing graph}

Let $G=(V,E)$ be an Eulerian digraph (connected) and a distinguished vertex $s$ of $G$ that is called \emph{sink}. Let $\edgerem{G}{s}$ be the graph $G$ in which the out-going arcs of $s$ have been deleted. Clearly $\edgerem{G}{s}$ has a global sink $s$. The Chip-firing game on $G$ with sink $s$ is the ordinary Chip-firing game that is defined on the graph $\edgerem{G}{s}$.

Let $\beta$ be the configuration defined by for every $v\in V\backslash \set{s}$, $\beta(v)=1$ if $(s,v) \in E$ and $\beta(v)=0$ otherwise. Since $G$ is Eulerian, $\beta \sim \textbf{0}$ (after firing $-1$ time every vertex, except the sink). Lemma \ref{weaker condition} implies the burning algorithm.

\begin{lem}\cite{Dha90}
\label{burning algorithm}
Configuration $c$ is recurrent if and only if $c=(c+\beta)^{\circ}$. Moreover if $c$ is recurrent then each vertex of $G$ except for the sink fires exactly once during any sequence of legal firings to reach the stabilization of $(c+\beta)$.
\end{lem}

\noindent Note that the configuration $c+\beta$ can be regarded as the configuration resulting from firing the sink in the configuration $c$. Lemma \ref{burning algorithm} allows to define the notion of firing graph that is originally from \cite{Sch10}.

\begin{defi}
Let $c$ be a recurrent configuration and $c+\beta=d_0\overset{w_1}{\to}d_1\overset{w_2}{\to}d_2\overset{w_3}{\to}d_3\cdots\overset{w_k}{\to}d_k$ a legal firing sequence of $c$ such that $d_k=c$. This sequence of legal firings can be presented by $(w_1,w_2,\cdots,w_k)$ since $d_i$ is completely defined by $w_1,w_2,\cdots,w_i$ for $i\geq 1$. Lemma \ref{burning algorithm} implies that $k=|V|-1$ and $\set{w_1,w_2,\cdots,w_k}=V\backslash \set{s}$. The graph $\mathcal{F}=(\mathcal{V},\mathcal{E})$ with $\mathcal{V}=V$ and $\mathcal{E}=\set{(s,w_i):(s,w_i)\in E}\cup \set{(w_i,w_j): i <j\text{ and }(w_i,w_j) \in E}$ is called a \emph{firing graph} of $c$.
\end{defi}

\begin{figure}
\centering
\subfloat[An Eulerian digraph]{\label{fig:im14}\includegraphics[bb=0 0 340 174,width=2.05in,height=1.05in,keepaspectratio]{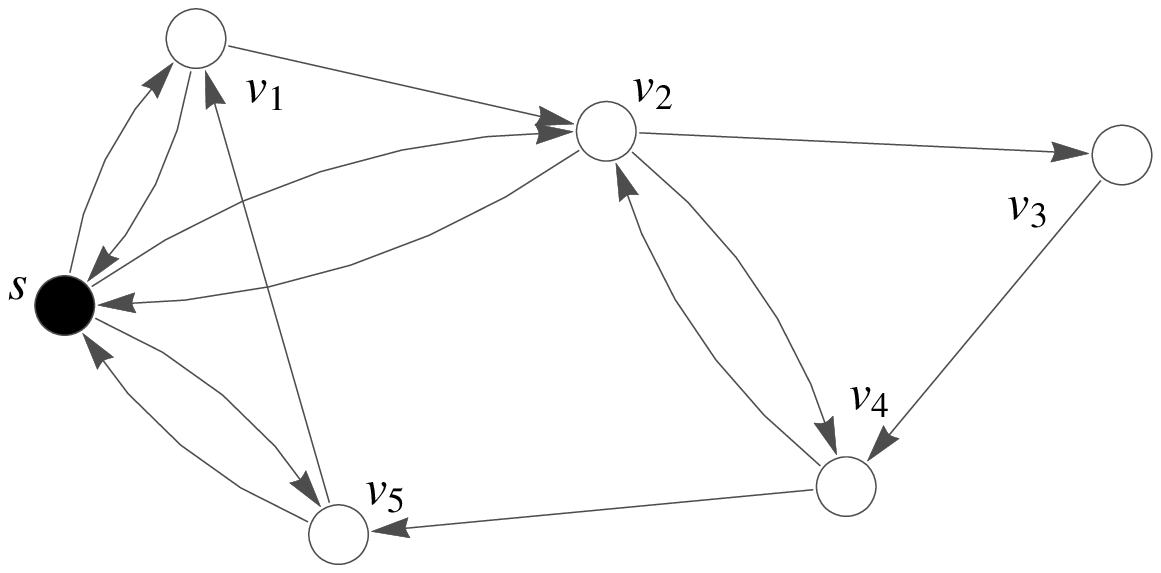}}\quad
\subfloat[A recurrent configuration]{\label{fig:im13}\includegraphics[bb=0 0 397 205,width=2.05in,height=1.06in,keepaspectratio]{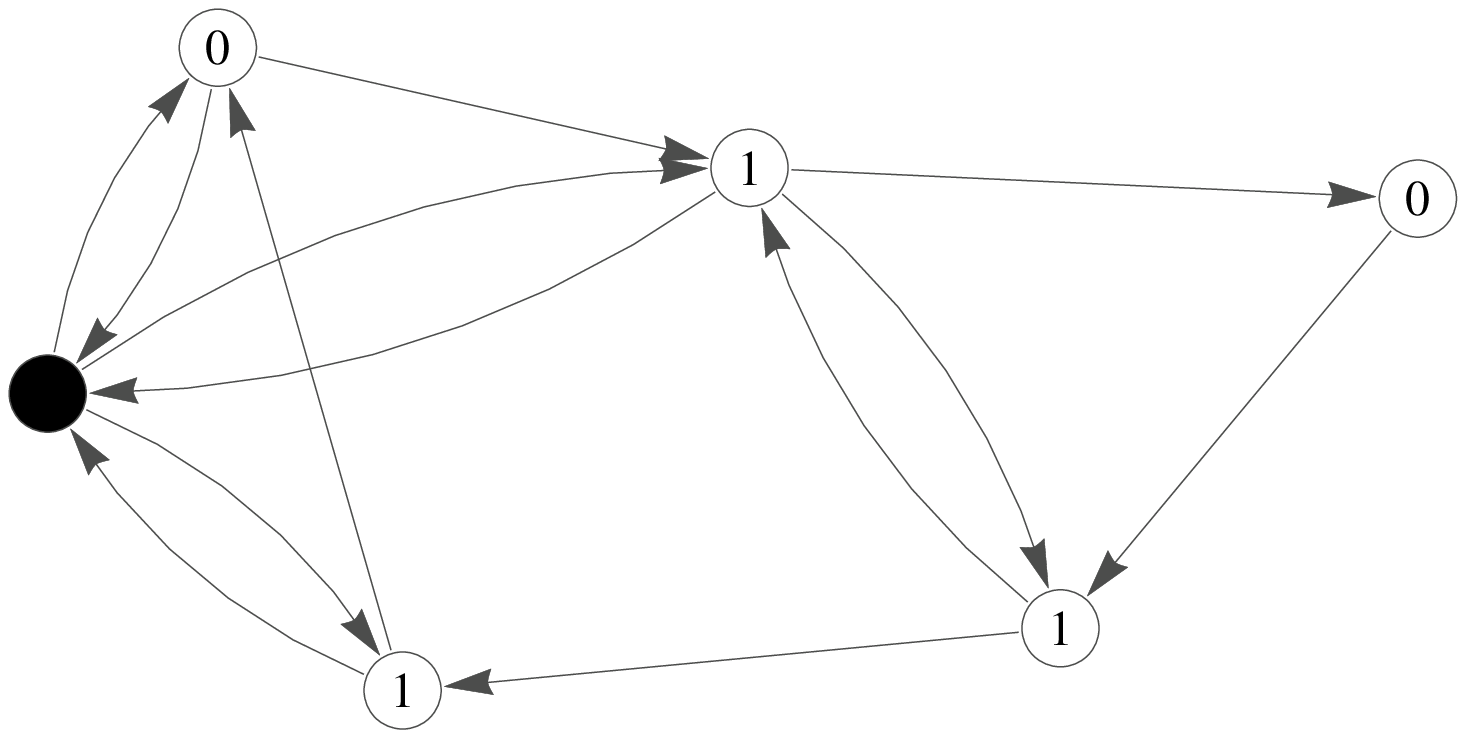}}\\
\subfloat[$c+\beta$]{\label{fig:im15}\includegraphics[bb=0 0 345 178,width=2.05in,height=1.06in,keepaspectratio]{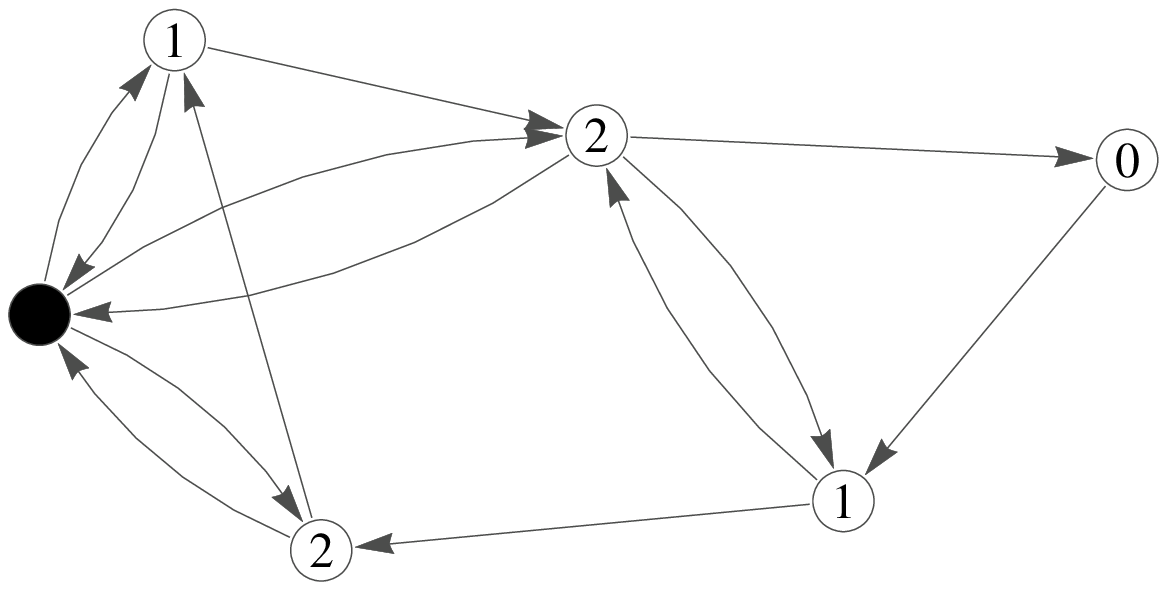}}\quad
\subfloat[Firing graph corresponding to the sequence $(v_5,v_1,v_2,v_4,v_3)$]{\label{fig:im16}\includegraphics[bb=0 0 340 174,width=2.05in,height=1.05in,keepaspectratio]{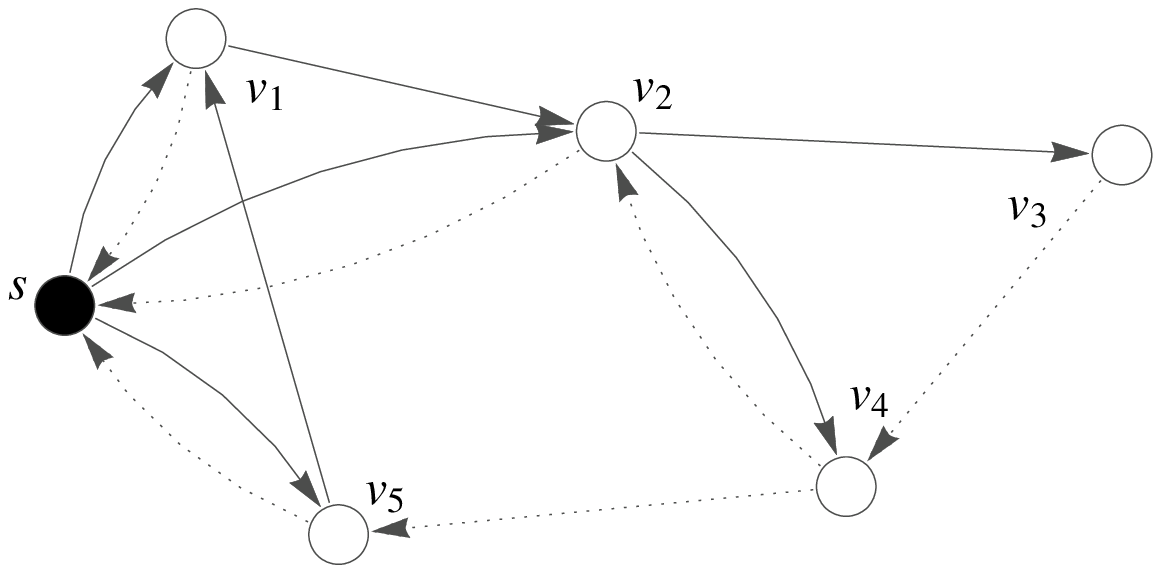}}
\caption{An example of firing graph}
\label{fig:im13141516}
\end{figure}

Figure \ref{fig:im14} presents an Eulerian digraph with the sink $s$ in black. Figure \ref{fig:im13} presents a recurrent configuration. The configuration $c+\beta$ is presented in Figure \ref{fig:im15}. Starting with the configuration $c+\beta$ we can fire consecutively the vertices $v_5,v_1,v_2,v_4,v_3$ of $V$ in this order to reach again $c$. With the legal firing sequence $(v_5,v_1,v_2,v_4,v_3)$ we have the firing graph that is presented by the undotted arcs in Figure \ref{fig:im16}. Note that legal firing sequences of $c+\beta$ are possibly not unique, so are firing graphs of $c$. In the next part we are going to study a kind of recurrent configurations that always have a unique firing graph.

\subsection{Minimal recurrent configurations and maximal acyclic arc sets}

In this subsection we work with the Chip-firing game on an Eulerian digraph $G=(V,E)$ with sink $s$. For two configurations $c'$ and $c$ we write $c'\leq c$ if $c'(v)\leq c(v)$ for every $v \in V\backslash \set{s}$. A recurrent configuration $c$ is \emph{minimal} if whenever $c'\neq c$ and $c'\leq c$, $c'$ is not recurrent. When $c$ has the minimum total number of chips over all recurrent configurations, we say that $c$ is \emph{minimum}. Let $\mathcal{M}$ be the set of all minimal recurrent configurations of the game.

Let $\mathcal{A}$ be the set of all maximal acyclic arc sets $A$ of $G$ such that $s$ is a unique sink of $A$. Note that maximal acyclic arc set can be considered as a generalization of acyclic orientation on undirected graphs. Figure \ref{fig:im18} shows such a maximal acyclic arc set of the Eulerian digraph shown in Figure \ref{fig:im14}. This subsection is devoted to showing that if a recurrent configuration $c$ is minimal, $c$ has a unique firing graph and the set of arcs of this firing graph is a maximal acyclic arc set. This gives a map from $\mathcal{M}$ to $\mathcal{A}$. Moreover we show that this map is a one-to-one correspondence between $\mathcal{M}$ and $\mathcal{A}$. The correspondence can be generalized easily to the case when $G$ has multi-arcs.

\begin{figure}
\centering
\includegraphics[bb=0 0 340 174,width=2.05in,height=1.05in,keepaspectratio]{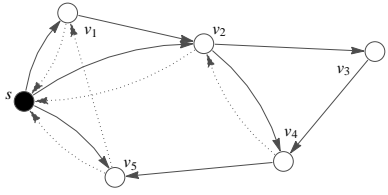}
\caption{A maximal acyclic arc set}
\label{fig:im18}
\end{figure}

When $G$ is an undirected graph, the correspondence  is exactly the one that was given in \cite{Sch10}. The correspondence in \cite{Sch10} deals with the case when $G$ has many sinks. However,  the many-sink case is not harder than the single-sink case since we can contract many sinks to a single sink, and consider the contracted graph. This subsection mainly focuses on showing a relation between $\mathcal{M}$ and $\mathcal{A}$, and not all results presented here are needed for the proof of the NP-hardness exposed in the next subsection. The following shows a basic relation between acyclic arc sets and recurrent configurations.

\begin{lem}
\label{from acyclic set to recurrent configuration}
Let $A$ be an acyclic arc set such that $s$ is a unique vertex of indegree $0$ in $G[A]$ and $A$ contains all vertices of $G$. Then the configuration $c$ defined by $c(v)=\outdegree{G}{v}-\indegree{G[A]}{v}$ for every $v \in V\backslash \set{s}$ is recurrent.
\end{lem}

\begin{proof}
Since $G[A]$ is acyclic, there is a linear order $v_0<v_1<v_2,\cdots<v_{|V|-1}$ on $V$ such that if $(v_i,v_j)\in A$ then $i <j$. Clearly $v_0=s$. The proof is completed by showing that $(v_1,v_2,\cdots,v_{|V(G)|-1})$ is a legal firing sequence of $c+\beta$. Since $v_1$ is an out-neighbor of $s$ in $G[A]$, we have $c(v_1)=\outdegree{G}{v_1}-1$, therefore it is active in $c+\beta$. Now by induction, suppose that $(v_1,v_2,\cdots,v_j)$ is a legal firing sequence of $c+\beta$, where $j<|V(G)|-1$. By firing consecutively the vertices $v_1,v_2,\cdots,v_j$ in this order we arrive at the configuration $c'$. It suffices to show that $v_{j+1}$ is active in $c'$. It is  clear that $v_{j+1}$ receives $\underset{0\leq i \leq j}{\sum} \nofedges{G}{v_i}{v_{j+1}}$ chips from its in-neighbors after all vertices $v_1,v_2,\cdots,v_j$ have been fired. Since $\underset{0\leq i \leq j}{\sum} \nofedges{G}{v_i}{v_{j+1}}\geq \indegree{G[A]}{v_{j+1}}$, the number of chips stored at $v_{j+1}$ in $c'$ is not less than $\outdegree{G}{v_{j+1}}$, therefore $v_{j+1}$ is active in $c'$. The claim follows.
\end{proof}

From the definition of firing graph, a recurrent configuration may have many firing graphs. However, the following implies that the numbers of arcs of those firing graphs have a lower bound that depends on the recurrent configuration.

\begin{lem}
\label{from recurrent configuration to acyclic set}
If $c$ is a recurrent configuration of $\edgerem{G}{s}$ then for every firing graph $\mathcal{F}=(\mathcal{V},\mathcal{E})$ of $c$, $s$ is a unique vertex of in-degree $0$ and $\mathcal{E}$ is an  acyclic arc set of $G$. Moreover, $\mathcal{F}$ is connected and for each $v \in V\backslash \set{s}$  we have $c(v) \geq \outdegree{G}{v}-\indegree{\mathcal{F}}{v}$.
\end{lem}

\begin{proof}
It follows immediately from the definition of firing graph that $s$ is a vertex of in-degree $0$ in $\mathcal{F}$ and $\mathcal{E}$ is an acyclic arc set. We show that there is no other vertex of in-degree $0$ in $\mathcal{F}$. Let $(v_1,v_2,\cdots,v_{|V|-1})$ be a legal firing sequence of $c+\beta$ that is used to construct $\mathcal{F}$. By convention $v_0=s$. For each $1 \leq i \leq |V|-1$ let $c'$ denote the configuration obtained from $c+\beta$ by firing consecutively the vertices $v_1,v_2,\cdots,v_{i-1}$. Since $v_i$ is not active in $c$ but active in $c'$, $v_i$ must receive some chips during this firing process. This implies that there is $j<i$ such that $(v_j,v_i) \in E$. It follows from the definition of firing graph that $(v_j,v_i) \in \mathcal{F}$, therefore $\indegree{\mathcal{F}}{v_i}\geq 1$. Since $\mathcal{F}$ is acyclic and has exactly one vertex of in-degree $0$, $\mathcal{F}$ is connected.

It remains to prove that for every $v \in V\backslash \set{s}$ we have $c(v)\geq \outdegree{G}{v}-\indegree{\mathcal{F}}{v}$. For every $1\leq i\leq |V|-1$ vertex $v_{i}$ receives $\indegree{\mathcal{F}}{v_i}$ chips from its in-neighbors when all vertices $v_1,v_2,\cdots,v_{i-1}$ have been fired. At this point $v_{i}$ is active, therefore $c(v_i)\geq \outdegree{G}{v_i}-\indegree{\mathcal{F}}{v_i}$.
\end{proof}

The notion of firing graph gives a map from $\mathcal{M}$ to $\mathcal{A}$ that is shown in the following.

\begin{lem}
\label{the corollary of the restriction on firing graph}
Let $c\in \mathcal{M}$ and $\mathcal{F}=(\mathcal{V},\mathcal{E})$ a firing graph of $c$. Then $c(v)=\outdegree{G}{v}-\indegree{\mathcal{F}}{v}$ for every $v\in V\backslash \set{s}$ and $\mathcal{E} \in \mathcal{A}$. Moreover, the configuration $c$ contains $|E|-\outdegree{G}{s}-|\mathcal{E}|$ chips.
\end{lem}


\begin{proof}
Let $c'$ be the configuration defined by $c'(v)=\outdegree{G}{v}-\indegree{\mathcal{F}}{v}$ for every $V\backslash \set{s}$. By Lemma \ref{from acyclic set to recurrent configuration} $c'$ is a recurrent configuration. It follows from Lemma \ref{from recurrent configuration to acyclic set} that $c'\leq c$. Since $c$ is minimal, we have $c'=c$, therefore $c(v)=\outdegree{G}{v}-\indegree{\mathcal{F}}{v}$ for every $v \in V\backslash \set{s}$.

To prove $\mathcal{E}\in \mathcal{A}$, we assume otherwise that there is $A\in \mathcal{A}$ such that $\mathcal{E}\subsetneq A$ (from Lemma \ref{from recurrent configuration to acyclic set} we know that $\mathcal{E}$ is an acyclic arc set, hence it is not maximal). Let $c''$ be the configuration defined by $c''(v)=\outdegree{G}{v}-\indegree{G[A]}{v}$ for every $v \in V\backslash \set{s}$. Let $(u,u') \in A\backslash \mathcal{E}$. Clearly $\indegree{G[A]}{u'}>\indegree{\mathcal{F}}{u'}$, therefore $c''(u')<c(u')$. It implies that $c''\neq c$ and $c''\leq c$, a contradiction to the fact that $c \in \mathcal{M}$.

The number of chips $c$ contains is $\underset{v\neq s}{\sum} c(v)=\underset{v\neq s}{\sum} (\outdegree{G}{v}-\indegree{\mathcal{F}}{v})=\underset{v\in V}{\sum}\outdegree{G}{v}-\outdegree{G}{s}-|\mathcal{E}|=|E|-\outdegree{G}{s}-|\mathcal{E}|$. The second statement follows.
\end{proof}

For two non-repeated sequences $\mathfrak{f}=(v_1,v_2,\cdots,v_{|V|-1})$ and $\mathfrak{g}=(w_1,w_2,\cdots,w_{|V|-1})$ of the vertices in $V\backslash \set{s}$, $\pre{\mathfrak{f}}{\mathfrak{g}}$ denotes the maximum integer $k$ such that for every $i$ satisfying $1 \leq i\leq k$, we have $v_k=w_k$. Note that if $v_1 \neq w_1$ then $\pre{\mathfrak{f}}{\mathfrak{g}}=0$. The following shows that there is a well-defined and injective map from $\mathcal{M}$ to $\mathcal{A}$.

\begin{lem}
\label{uniqueness of firing graph}
For every $c \in \mathcal{M}$, $c$ has exactly one firing graph.
\end{lem}

\begin{proof} Let $\mathfrak{f}_1=(v_1,v_2,\cdots,v_{|V|-1})$ and $\mathfrak{f}_2=(w_1,w_2,\cdots,w_{|V|-1})$ be two different legal firing sequences of $c+\beta$. Let $j$ denote $\pre{\mathfrak{f}_1}{\mathfrak{f}_2}$ and $\mathfrak{f}'=(v_1,v_2,\cdots,v_j,w_{j+1},v_{j+1},v'_{j+3},v'_{j+4},\cdots,v'_{|V|-1})$ the sequence of vertices of $G$, where $(v'_{j+3},v'_{j+4},\cdots,v'_{|V|-1})$ is the sequence $(v_{j+2},\cdots,v_{|V|-1})$ with $w_{j+1}$ deleted. Clearly, $\mathfrak{f}'$ is also a legal firing sequence of $c+\beta$. Let $\mathcal{F}_1=(\mathcal{V}_1,\mathcal{E}_1)$ and $\mathcal{F}'=(\mathcal{V}',\mathcal{E}')$ denote the firing graphs of $c$ with respect to $\mathfrak{f}_1$ and $\mathfrak{f}'$, respectively.

We claim that $\mathcal{F}_1=\mathcal{F}'$. Lemma \ref{the corollary of the restriction on firing graph} implies that $|\mathcal{E}_1|=|\mathcal{E}'|=\underset{v \in V\backslash \set{s}}{\sum} \outdegree{G}{v}-\underset{v \in V\backslash \set{s}}{\sum} c(v)$.  Hence it suffices to prove that $\mathcal{E}_1\backslash \mathcal{E}'=\emptyset$. We assume otherwise that $\mathcal{E}_1\backslash \mathcal{E}'\neq\emptyset$. Let $k$ denote the integer such that $w_{j+1}=v_k$. Note that $k >j+1$. It follows from the definition of firing graph that $\mathcal{E}_1\backslash \mathcal{E}'=\set{(v_i,v_k) \in E: j+1\leq i \leq k-1}$. Let $X=\set{(v_i,v_k): (v_i,v_k) \in \mathcal{F}'}$ and $Y=\set{(v_i,v_k): (v_i,v_k) \in \mathcal{F}_1}$. Since $\mathfrak{f}'$ can be viewed as $\mathfrak{f}_1$ in which $v_k$ has been moved backward, we have $X\subseteq Y$. It follows from $\mathcal{E}_1\backslash \mathcal{E}'\neq \emptyset$ that $X\subsetneq Y$, therefore $\indegree{\mathcal{F}'}{v_k}<\indegree{\mathcal{F}_1}{v_k}$, a contradiction to the assertion of  Lemma \ref{the corollary of the restriction on firing graph}.

Let $\mathcal{F}_2$ denote the firing graph of $c$ constructed by $\mathfrak{f}_2$. The proof is completed by showing that $\mathcal{F}_1=\mathcal{F}_2$. Let $\delta=(\delta_1,\delta_2,\cdots,\delta_{|V|-1})$ be a legal firing sequence of $c+\beta$ such that the firing graph constructed by $\delta$ is the same as $\mathcal{F}_1$ and $\pre{\delta}{\mathfrak{f}_2}$ is maximum. We are going to show that $\delta=\mathfrak{f}_2$. Let $p$ denote $\pre{\delta}{\mathfrak{f}_2}$. If $\delta\neq \mathfrak{f}_2$ then $p <|V|-1$. Let $\delta'$ denote the sequence $(\delta_1,\delta_2,\cdots,\delta_p,w_{p+1},\delta_{p+1},u_{p+3},u_{p+4},\cdots,u_{|V|-1})$ of vertices of $G$, where $(u_{p+3},u_{p+4},\cdots,u_{|V|-1})$ is the sequence $(\delta_{p+2},\delta_{p+3},\cdots,\delta_{|V|-1})$ with the vertex $w_{p+1}$ deleted. The above claim implies that the firing graph of $c$ constructed by $\delta'$ is the same as the one constructed by $\delta$. It is clear that $\pre{\delta'}{\mathfrak{f}_2}>\pre{\delta}{\mathfrak{f}_2}$, a contradiction to the maximum of $\pre{\delta}{\mathfrak{f}_2}$
\end{proof}

For two non-repeated sequences $\mathfrak{f}=(v_1,v_2,\cdots,v_{|V|-1}),\mathfrak{g}=(w_1,w_2,\cdots,w_{|V|-1})$ of vertices in $V\backslash \set{s}$ we denote by $\inter{\mathfrak{f}}{\mathfrak{g}}$ the sequence $(v_1,v_2,\cdots,v_{k},w_{k+1},v_{k+1},v'_{k+3},v'_{k+4},\cdots,v'_{|V|-1})$, where $k=\pre{\mathfrak{f}}{\mathfrak{g}}$ and $(v'_{k+3},v'_{k+4},\cdots,v'_{|V|-1})$ is the sequence $(v_{k+2},v_{k+3},\cdots,v_{|V|-1})$ with the vertex $w_{k+1}$ deleted. It is easy to see that $\pre{\mathfrak{f}}{\mathfrak{g}}<\pre{\inter{\mathfrak{f}}{\mathfrak{g}}}{\mathfrak{g}}$.  Note that if $\mathfrak{f}$ and $\mathfrak{g}$ are two legal firing sequences of a configuration $c$, $\inter{\mathfrak{f}}{\mathfrak{g}}$ is also a legal firing sequence of $c$. The following result is the converse of Lemma \ref{the corollary of the restriction on firing graph}.

\begin{lem}
\label{from maximality to minimality}
Let $A\in \mathcal{A}$ and $\mathcal{F}$ denote $G[A]$. Then the configuration $c$ defined by $c(v)=\outdegree{G}{v}-\indegree{\mathcal{F}}{v}$ for every $v \in V\backslash \set{s}$ is a minimal recurrent configuration.
\end{lem}

\begin{proof}
For a contradiction we assume otherwise that $c$ is not minimal. There is $c' \in \mathcal{M}$ such that $c'\neq c$ and $c'\leq c$. Let $\mathcal{F}'$ be the firing graph of $c'$. By Lemma \ref{the corollary of the restriction on firing graph} we have $E(\mathcal{F}')\in \mathcal{A}$ and $\mathcal{F}' \neq \mathcal{F}$.

Since $A$ is acyclic, there is a non-repeated sequence $\mathfrak{f}_1=(v_1,v_2,\cdots,v_{|V|-1})$ of vertices in $V\backslash \set{s}$ such that if $(v_i,v_j) \in A$ then $i<j$. Clearly, $\mathfrak{f}_1$ is a legal firing sequence of $c+\beta$. Similarly, there is a non-repeated sequence $\mathfrak{f}_2=(w_1,w_2,\cdots,w_{|V|-1})$ of vertices $V\backslash \set{s}$ such that if $(w_i,w_j) \in E(\mathcal{F}')$ then  $i<j$. Clearly, $\mathfrak{f}_2$ is a legal firing sequence of $c'+\beta$. We define the sequence $\set{\mathfrak{g}_i}_{i \in \mathbb{N}}$ as follows
\begin{align*}
&\mathfrak{g}_0=\mathfrak{f}_1\\
&\mathfrak{g}_{i+1}=\inter{\mathfrak{g}_i}{\mathfrak{f}_2}, i\geq 0
\end{align*}
Let $p$ be the minimum integer such that $\mathfrak{g}_p=\mathfrak{f}_2$. Note that for every $i \geq p$, $\mathfrak{g}_i=\mathfrak{f}_2$. Since $\mathcal{F} \neq \mathcal{F}'$, there is a minimum integer $q<p$ such that the firing graph constructed by $\mathfrak{g}_q=(\delta_1,\delta_2,\cdots,\delta_{|V|-1})$ is distinct from the firing graph constructed by $\mathfrak{g}_{q+1}$. Let $k=\pre{\mathfrak{g}_q}{\mathfrak{f}_2}$ and $l$ be the integer such that $\delta_l=w_{k+1}$. The firing graphs constructed by $\mathfrak{g}_q$ and $\mathfrak{g}_{q+1}$ are denoted by $\mathcal{G}_1$ and $\mathcal{G}_2$, respectively.

We claim that for every $k+1\leq i \leq l-1$ we have $(\delta_i,\delta_l) \not \in E$. For a contradiction we assume otherwise. By a similar argument as in the proof of Lemma \ref{uniqueness of firing graph}, the set of arcs of $\mathcal{G}_2$ whose head $\delta_l$ is a subset of the set of arcs of $\mathcal{G}_1$ whose head $\delta_l$. The assumption implies that there is an arc $e\in E$ such that $e\in \mathcal{G}_1$ and $e \not \in \mathcal{G}_2$, therefore $\indegree{\mathcal{G}_2}{\delta_l}<\indegree{\mathcal{G}_1}{\delta_l}$. Since $\pre{\mathfrak{g}_i}{\mathfrak{f}_2}<\pre{\mathfrak{g}_{i+1}}{\mathfrak{f}_2}$ for every $0\leq i \leq p-1$, $\indegree{\mathcal{G}_2}{\delta_l}$ is equal to the in-degree of $\delta_l$ in the firing graph constructed by $\mathfrak{g}_p=\mathfrak{f}_2$, namely $\mathcal{F}'$. It follows that $\indegree{\mathcal{F}}{\delta_l}=\indegree{\mathcal{G}_1}{\delta_l}>\indegree{\mathcal{G}_2}{\delta_l}=\indegree{\mathcal{F}'}{\delta_l}$, therefore $c(\delta_l)<c'(\delta_l)$, a contradiction to the fact that $c'\leq c$.

Since $E(\mathcal{G}_1)\backslash E(\mathcal{G}_2)=\set{(\delta_i,\delta_l) \in E: k+1\leq i\leq l-1}$, it follows from the above claim that $E(\mathcal{G}_1)\backslash E(\mathcal{G}_2)=\emptyset$, therefore $E(\mathcal{G}_1)\subsetneq E(\mathcal{G}_2)$. The choice of $q$ implies that $E(\mathcal{G}_1)=A$, a contradiction to the fact that $A$ is a maximal acyclic arc set.
\end{proof}

The following is the main result of this subsection.

\begin{theo}
\label{one-to-one correspondence}
Let $\mathcal{F}_c$ denote the firing graph of $c$, the map from $\mathcal{M}$ to $\mathcal{A}$ defined by $c \mapsto \mathcal{F}_c$ is bijective.
\end{theo}

\begin{proof}
Lemma \ref{the corollary of the restriction on firing graph} and Lemma \ref{uniqueness of firing graph} imply that the map is well-defined and injective. Lemma \ref{from maximality to minimality} implies the surjectivity.
\end{proof}

We end this subsection with an interesting property of the Chip-firing game on Eulerian digraphs

\begin{prop}
\label{minrec-independence}
The number of minimum recurrent configurations is independent of the choice of sink.
\end{prop}

\begin{proof}
Theorem \ref{one-to-one correspondence} and Lemma \ref{the corollary of the restriction on firing graph} imply that the map $c\mapsto \mathcal{F}_c$ induces a map from the minimum recurrent configurations to the maximum acyclic arc sets of $G$ in $\mathcal{A}$. Therefore the number of minimum recurrent configurations is equal to the number of maximum acyclic arc sets of $G$ in $\mathcal{A}$. It follows from Proposition \ref{cardinality-independence} that the number of maximum acyclic arc sets of $G$ in $\mathcal{A}$ is independent of the choice of sink, so is the number of minimum recurrent configurations.
\end{proof}

Proposition \ref{minrec-independence} states that the number of minimum recurrent configurations is characteristic of the digraph itself.

\subsection{NP-hardness of minimum recurrent configuration problem}

In this subsection we study the computational complexity of the following problem
\begin{center}
\begin{tabular}{|l|}
\hline
\textbf{MINREC problem}\\
\text{}\\
\textbf{Input:} A graph $G$ with a global sink.\\
\textbf{Output:} Minimum total number of chips of a recurrent configuration of $G$.\\
\hline
\end{tabular}
\end{center}
If the input graphs are restricted to undirected graphs $G$ with a sink $s$, the problem can be solved in polynomial time since all minimal recurrent configurations have the same total number of chips, namely $\frac{E(G)}{2}$. Nevertheless, the problem is NP-hard for general digraphs. In particular, we show that the problem is NP-hard when the input graphs are restricted to Eulerian digraphs. 
\begin{center}
\begin{tabular}{|l|}
\hline
\textbf{EMINREC problem}\\
\text{}\\
\textbf{Input: } An Eulerian digraph $G$ with a sink $s$.\\
\textbf{Output: } Minimum total number of chips of a recurrent configuration of $G$.\\
\hline
\end{tabular}
\end{center}

\begin{theo}
The $EMINREC$ problem is NP-hard, so is the MINREC problem.
\end{theo}

\begin{proof}
Let $G$ be an Eulerian digraph with sink $s$. Let $k$ be the maximum number of arcs of a feedback arc set of $G$ and $k'$ be the minimum number of chips of a recurrent configuration of $G$. Since the EMINFAS problem is NP-hard, the proof is completed by showing that $k+k'=\underset{v \in V\backslash \set{s}}{\sum}\outdegree{G}{v}$.

By Theorem \ref{unique-source property} there is an acyclic arc set $A$ of $G$ such that $|A|=k$ and $s$ is a unique vertex of indegree $0$ in $G[A]$. Lemma \ref{from acyclic set to recurrent configuration} implies that the configuration $c$ defined by $c(v)=\outdegree{G}{v}-\indegree{G[A]}{v}$ for every $v \in V\backslash \set{s}$ is recurrent. Clearly $k+\underset{v \in V\backslash\set{s}}{\sum}c(v)=\underset{v \in V\backslash \set{s}}{\sum}\outdegree{G}{v}$ and $k+k'\leq \underset{v \in V\backslash \set{s}}{\sum} \outdegree{G}{v}$ since $G$ is Eulerian.

It remains to prove that $k+k'\geq \underset{v\in V\backslash \set{s}}{\sum}\outdegree{G}{v}$. Let $\bar{c}$ be a recurrent configuration such that $\underset{v \in V\backslash \set{s}}{\sum} \bar{c}(v)=k'$. Let $\mathcal{F}$ be a firing graph of $\bar{c}$. Lemma \ref{from recurrent configuration to acyclic set} implies that $\bar{c}(v)\geq \outdegree{G}{v}-\indegree{\mathcal{F}}{v}$ for every $v \in V\backslash \set{s}$, therefore $k+k'\geq \underset{v \in V\backslash \set{s}}{\sum} \bar{c}(v)+|E(\mathcal{F})|\geq \underset{v \in V\backslash \set{s}}{\sum} \outdegree{G}{v}$.
\end{proof}

Note that it follows directly from \cite{Sta91} that the EMINFAS problem restricted to planar Eulerian digraphs is solvable in polynomial time, so is the EMINREC problem. This class of graphs is pretty big since it contains planar undirected graphs.
\section{Conclusion and perspectives}

In this paper we pointed out a close relation between the MINFAS problem and the MINREC problem. The important consequence of this relation is the NP-hardness of the MINREC problem. It would be interesting to investigate classes of graphs that are situated strictly between the class of undirected graphs and the class of Eulerian digraphs, for which the MINFAS and MINREC problems are solvable in polynomial time. We discuss here about such a class.

It follows from Theorem \ref{unique-source property} that to compute the maximum number of arcs of an acyclic arc set of an Eulerian digraph,  we can restrict to the acyclic arc sets that satisfy the condition in Theorem \ref{unique-source property}. With different choices of $s$ we have different sets of maximal acyclic arc sets. One would prefer to choose a vertex $s$ such that all maximal acyclic arc set have the same number of arcs since a maximal acyclic arc set can be computed quickly, therefore a maximum acyclic arc set. Figure \ref{fig:im19} shows an Eulerian digraph. If $v_1$ is chosen, we have exactly one maximal acyclic arc set that is shown in Figure \ref{fig:im20}. If $v_2$ is chosen, we have exactly two maximal acyclic arc sets with different sizes. Thus one computes easily a maximum acyclic arc set if $v_1$ is chosen.

\begin{figure}[!h]
\centering
\subfloat[An Eulerian digraph]{\label{fig:im19}\includegraphics[bb=0 0 373 165,width=2.05in,height=0.905in,keepaspectratio]{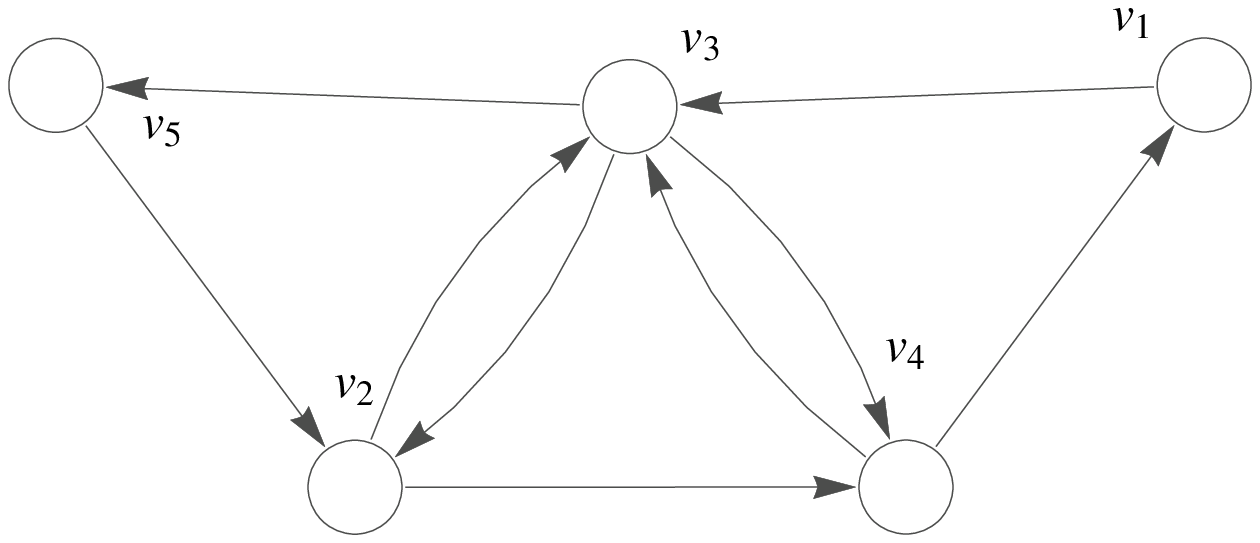}}\quad
\subfloat[A maximal acyclic arc set with respect to $v_1$]{\label{fig:im20}\includegraphics[bb=0 0 297 133,width=2.05in,height=0.914in,keepaspectratio]{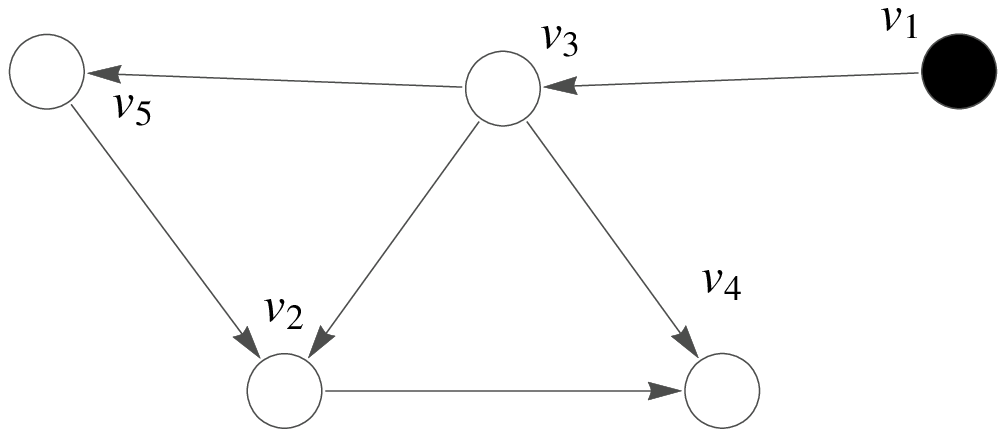}}\\
\subfloat[Maximal acyclic arc sets with respect to $v_2$]{\label{fig:im23}\includegraphics[bb=0 0 461 107,width=4.09in,height=0.945in,keepaspectratio]{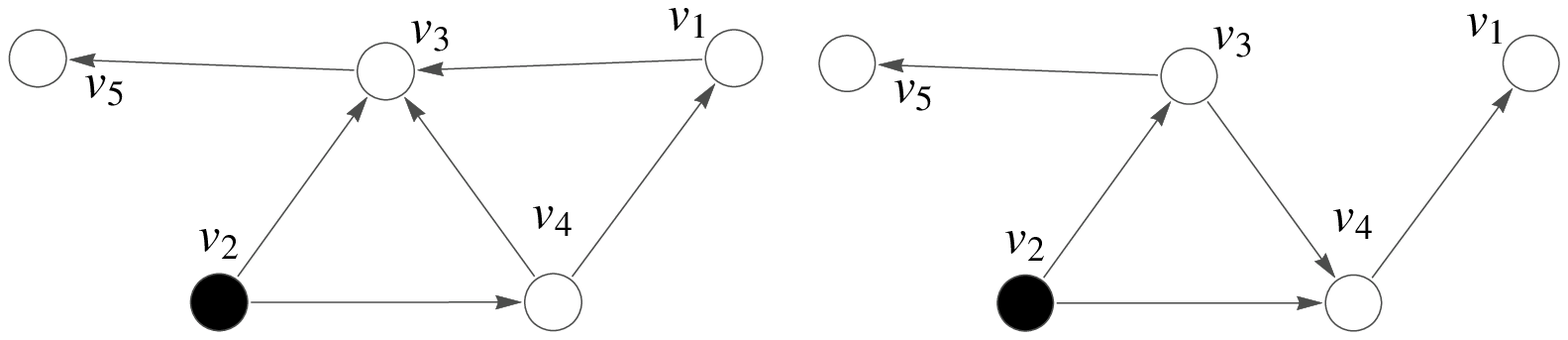}}
\caption{Maximal acyclic arc sets with different choices of $s$}
\label{fig:im192023}
\end{figure}

Note that there are many Eulerian digraphs in each of which there is no vertex $s$ that satisfies this good property. By an experimental observation we see that the class of Eulerian digraphs, for which at least one vertex $s$ has the property, is rather large. However, a characterization for this class  of graphs, on which the MINFAS problem is polynomial, is unknown and remains to be done. In addition, the observation also provides a heuristic algorithm for the EMINFAS problem. It is interesting to investigate the properties of this algorithm.

We also presented in this paper a number of interesting properties of feedback arc sets and recurrent configurations of the Chip-firing game on Eulerian digraphs. One of the most interesting properties is the one in Proposition \ref{minrec-independence}. We propose here an open question that is currently in our interests for further investigations: Is there any stronger result for Proposition \ref{minrec-independence} on Eulerian digraphs, and on digraphs? We believe that the results we presented in this paper can be generalized to general digraphs.

\section*{Acknowledgments}

We would like to thank Holger-Frederik Robert Flier for noticing us that the NP-hardnesss of the MINFAS problem on Eulerian multi-digraphs has been discovered in his PhD thesis. We  would also like to thank him for the useful discussions.

K\'evin Perrot\\
Universit\'e de Lyon - LIP (UMR 5668 CNRS-ENS de Lyon-Universit\'e Lyon 1)\\
46 all\'ee d'Italie 69364 Lyon Cedex 7-France\\
Universit\'e de Nice Sophia Antipolis - Laboratoire I3S (UMR 6070 CNRS)\\
2000 route des Lucioles, BP 121, F-06903 Sophia Antipolis Cedex, France\\
email: kevin.perrot@ens-lyon.fr\\
\text{}\\
Trung Van Pham\\
Department of Mathematics of Computer Science\\
Vietnam Institute of Mathematics\\
18 Hoang Quoc Viet Road, Cau Giay District, Hanoi, Vietnam\\
email: pvtrung@math.ac.vn

\end{document}